\documentclass[11pt]{article}

\usepackage{fullpage}

\usepackage{dsfont}
\usepackage{latexsym}
\usepackage{natbib}
\usepackage{amsmath}
\usepackage{amsthm}
\usepackage{amssymb}
\usepackage{amsfonts}
\usepackage{aliascnt}
\usepackage{graphicx}
\usepackage{enumerate}

\usepackage{algorithmic}
\usepackage{algorithm}
\usepackage[pdfpagelabels,pdfpagemode=None]{hyperref}

\newcommand{\argmin}{\operatorname{arg\,min}}

\newtheorem{theorem}{Theorem}%

\newtheorem{lemma}{Lemma}%

\newtheorem{claim}{Claim}%

\newtheorem{corollary}{Corollary}%

\newtheorem{proposition}{Proposition}%

\theoremstyle{definition}
\newtheorem{definition}{Definition}

\DeclareMathOperator*{\E}{E}%

\usepackage{ifthen}

\newcommand{\Xcomment}[1]{}

\newcommand{\prob}[2][]{\text{\bf Pr}\ifthenelse{\not\equal{}{#1}}{_{#1}}{}\!\left[#2\right]}
\newcommand{\expect}[2][]{\text{\bf E}\ifthenelse{\not\equal{}{#1}}{_{#1}}{}\!\left[#2\right]}
\newcommand{\given}{\;\middle|\;}

\newcommand{\distribution}[2][]{\text{\bf Dist}\ifthenelse{\not\equal{}{#1}}{_{#1}}{}\!\left[#2\right]}

\newcommand{\type}{t}
\newcommand{\types}{{\mathbf \type}}
\newcommand{\typespace}{T}
\newcommand{\typespaces}{{\mathbf \typespace}}

\newcommand{\distover}[1]{{\Delta{(#1)}}}
\newcommand{\expost}[1]{\hat{#1}}
\newcommand{\interim}[1]{#1}

\newcommand{\epallocs}{{\expost{\allocs}}}
\newcommand{\epalloc}{{\expost{\alloc}}}
\newcommand{\epoutcomes}{{\expost{\outcomes}}}
\newcommand{\epoutcome}{{\expost{\outcome}}}
\newcommand{\outcomespace}{W}

\newcommand{\outcome}{w}
\newcommand{\outcomes}{{\mathbf \outcome}}

\newcommand{\distoutcomespace}{\outcomedistspace}
\newcommand{\outcomespaces}{{\mathbf \outcomespace}}
\newcommand{\outcomedistspace}{{\distover{\outcomespace}}}
\newcommand{\outcomedistspaces}{{\distover{\outcomespaces}}}
\newcommand{\agent}{i}
\newcommand{\agind}[1][\agent]{_{#1}}

\newcommand{\toutcome}{\outcome}

\newcommand{\intoutcome}{\interim{\outcome}}
\newcommand{\intoutcomes}{\interim{\outcomes}}

\newcommand{\numservice}{m}

\newcommand{\alloc}{x}
\newcommand{\allocs}{{\mathbf \alloc}}

\newcommand{\intalloc}{\interim{\alloc}}
\newcommand{\intallocs}{\interim{\allocs}}
\newcommand{\talloc}{\alloc}

\newcommand{\util}{u}

\newcommand{\dens}{f}
\newcommand{\denss}{\mathbf \dens}

\DeclareMathOperator{\ALLOC}{Alloc}
\DeclareMathOperator{\PAYMENT}{Payment}
\DeclareMathOperator{\REV}{Rev}
\DeclareMathOperator{\RULE}{Outcome}

\newcommand{\reals}{{\mathbb R}}
\newcommand{\posreals}{\reals_+}

\newcommand{\Dist}{\mathcal{D}}
\newcommand{\Inter}[1]{{#1}} 

\newcommand{\Angles}[1]{{\langle{#1}\rangle}}

\newcommand{\Src}{{\Angles{\operatorname{\textsc{Src}}}}}
\newcommand{\Snk}{{\Angles{\operatorname{\textsc{Snk}}}}}
\newcommand{\Node}[2]{{\Angles{{#1},{#2}}}}
\newcommand{\SNode}[1]{{\Angles{{#1}}}}

\newcommand{\Agents}{N}
\newcommand{\NumTypes}{D}

\newcommand{\Range}[2]{\AutoAdjust{\{}{{#1},\ldots,{#2}}{\}}}

\newcommand{\Cut}{\operatorname{\textsc{Cut}}}
\newcommand{\Reroute}{\operatorname{\textsc{Reroute}}}

\newcommand{\InBrackets}[1]{\AutoAdjust{[}{#1}{]}}
\newcommand{\Ex}[2][]{\operatorname{\mathbf E}_{#1}\InBrackets{#2}}
\newcommand{\Prx}[2][]{\operatorname{\mathbf{Pr}}_{#1}\InBrackets{#2}}

\newcommand{\Universe}{\mathit{U}}
\newcommand{\USubset}{S}
\newcommand{\SubFun}{\mathcal{F}}
\newcommand{\PolyMat}[1]{\mathit{P}_{#1}}
\newcommand{\PolyVec}{y}
\newcommand{\PolyVecDelta}{\hat{\PolyVec}}
\newcommand{\Elem}{s}

\newcommand{\Mat}{\mathcal{M}}
\newcommand{\IndepSets}{\mathcal{I}}
\newcommand{\IndepSet}{I}
\newcommand{\Rank}[2][]{r^{#1}_{#2}}

\newcommand{\TypeSpace}{T}
\newcommand{\TypeSpaces}{{\mathbf \TypeSpace}}
\newcommand{\UnionTypeSpace}{{\TypeSpace_{\Agents}}}

\newcommand{\AllocSpace}{\mathbb{X}}
\newcommand{\InAllocSpace}{{\Inter{\AllocSpace}}}

\newcommand{\SeqAllocSpace}{\mathbb{S}}

\newcommand{\Type}{t}
\newcommand{\Types}{{\mathbf \Type}}
\newcommand{\TypeVar}{\tau}
\newcommand{\NonAugType}{\TypeVar^*}
\newcommand{\TypeSubset}{S}

\newcommand{\TightSets}{\mathds{S}}

\newcommand{\VertexSet}[1]{{\operatorname{\textsc{Vertex}}(#1)}}
\newcommand{\Vertex}[2]{{\operatorname{\textsc{Vertex}}(#2, {#1})}}

\newcommand{\AgentSubset}{{\Agents'}}
\newcommand{\DistProb}{f}
\newcommand{\DistProbs}{{\mathbf f}}
\newcommand{\DistOracle}{g}

\newcommand{\DistDP}[3]{G_{#2}^{#1}}

\newcommand{\Alloc}{x}
\newcommand{\Allocs}{{\mathbf \Alloc}}

\newcommand{\InAlloc}{{\Inter{\Alloc}}}
\newcommand{\InAllocs}{{\Inter{\Allocs}}}

\newcommand{\RvAlloc}[1]{\epalloc^{#1}}

\newcommand{\IVec}[1]{{\mathbf 1}_{#1}}
\newcommand{\RandRound}{\operatorname{\textsc{RandRound}}}

\newcommand{\RandBits}{b}

\newcommand{\SeqAlloc}{y}
\newcommand{\SeqAllocs}{y}
\newcommand{\SeqTrans}{z}
\newcommand{\SeqTranses}{\SeqTrans}
\newcommand{\SeqTable}{\pi}
\newcommand{\RCap}{\operatorname{\textsc{ResCap}}}

\newcommand{\PreAlloc}{\overline{x}}
\newcommand{\PreAllocSpace}{\overline{{\mathds X}}}

\newcommand{\Rev}{\operatorname{Rev}}

\newcommand{\AutoAdjust}[3]{\mathchoice{ \left #1 #2  \right #3}{#1 #2 #3}{#1 #2 #3}{#1 #2 #3} }

\newcommand{\Abs}[1]{{\AutoAdjust{\lvert}{#1}{\rvert}}}

\newcommand{\Reals}{\mathds{R}}
\newcommand{\PosReals}{\Reals_+}

\newcommand{\mdval}[2]{#1_{#2}}   
\newcommand{\mdalloc}[2]{#1_{#2}} 
\newcommand{\mdprice}[1]{#1_p}    

\newcommand{\bugbug}[1]{#1_b}
\newcommand{\bugval}[1]{#1_v}
\newcommand{\bugalloc}[1]{#1_x}
\newcommand{\bugprob}[1]{#1_\rho}
\newcommand{\bugprice}[1]{#1_p}

\newcommand{\RElem}{\pi}

\newcommand{\Ranking}{\pi}

\title{Bayesian Optimal Auctions \\ via Multi- to Single-agent Reduction}

\author{\hspace*{1in}
Saeed Alaei\thanks{email: \texttt{saeed@cs.umd.edu}, Dept.\ of Computer Science, University of Maryland, College Park, MD 20742.
Partially supported by ONR YIP grant N000141110662. Part of this work was done when the author was visiting Northwestern
University.  } \and Hu Fu\thanks{email: \texttt{hufu@cs.cornell.edu}. Dept.\ of Computer Science, Cornell University, Ithaca, NY
14853.  Supported by NSF grants CCF-0643934 and AF-0910940.  Part of this work was done when the author was visiting Northwestern
University.} \and Nima Haghpanah\thanks{email: \texttt{nima.haghpanah@gmail.com}. Dept.\ Electrical Engineering \& Computer
Science, Northwestern University, Evanston, IL 60201.} \hspace*{1in} \and Jason Hartline\thanks{email:
\texttt{hartline@eecs.northwestern.edu}. Dept.\ Electrical Engineering \& Computer Science, Northwestern University, Evanston, IL
60208.} \and Azarakhsh Malekian\thanks{email: \texttt{azarakhshm@gmail.com}. Dept.\ Electrical Engineering \& Computer Science,
Massachusetts Institute of Technology, Cambridge , MA 02139.} }

\begin{document}
\maketitle

\begin{abstract}
We study an abstract optimal auction problem for a single good or
service.  This problem includes environments where agents have
budgets, risk preferences, or multi-dimensional preferences over
several possible configurations of the good (furthermore, it allows an
agent's budget and risk preference to be known only privately to the
agent). These are the main challenge areas for auction theory.  A
single-agent problem is to optimize a given objective subject to a
constraint on the maximum probability with which each type is
allocated, a.k.a., an allocation rule.  Our approach is a reduction
from multi-agent mechanism design problem to collection of
single-agent problems.  We focus on maximizing revenue, but our
results can be applied to other objectives (e.g., welfare).

An optimal multi-agent mechanism can be computed by a linear/convex
program on interim allocation rules by simultaneously optimizing several
single-agent mechanisms  subject to joint feasibility of the
allocation rules.  For single-unit auctions, Border \citeyearpar{B91}
showed that the space of all jointly feasible interim allocation rules
for $n$~agents is a $\NumTypes$-dimensional convex polytope which can
be specified by $2^\NumTypes$ linear constraints, where $\NumTypes$ is
the total number of all agents' types.  Consequently, efficiently
solving the mechanism design problem requires a separation oracle for
the feasibility conditions and also an algorithm for ex-post
implementation of the interim allocation rules.  We show that the
polytope of jointly feasible interim allocation rules is the
projection of a higher dimensional polytope which can be specified by
only $O(\NumTypes^2)$ linear constraints.  Furthermore, our proof
shows that finding a preimage of the interim allocation rules in the
higher dimensional polytope immediately gives an ex-post
implementation.

We generalize Border's result to the case of $k$-unit and matroid
auctions.  For these problems we give a separation-oracle based
algorithm for optimizing over feasible interim allocation rules and a
randomized rounding algorithm for ex post implementation.  These ex
post implementations have a simple form; they are randomizations over
simple greedy mechanisms.  Given a ordered subset of agent types, such
a greedy mechanisms serves types in the specified order.

\end{abstract}

\section{Introduction}

Classical economics and game theory give fundamental characterizations
of the structure of competitive behavior.  For instance, Nash's
\citeyearpar{N51} theorem shows that mixed equilibrium gives a complete
description of strategic behavior, and the Arrow-Debreu \citeyearpar{AD54}
theorem shows the existence of market clearing prices in multi-party
exchanges.  In these environments computational complexity has offered
further perspective.  In particular, mixed equilibrium in general
games can be computationally hard to find \citep{CD06,DGP09}, whereas market
clearing prices are often easy to find \citep{DPSV08,J04}.  In this paper we
investigate an analogous condition for auction theory due to Kim
\citet{B91}, give a computationally constructive generalization that
further illuminates the structure of auctions, and thereby show that
the theory of optimal auctions is tractable.

Consider an abstract optimal auction problem.  A seller faces a set of
agents.  Each agent desires service and there may be multiple ways to
serve each agent (e.g., when renting a car, you can get a GPS or not,
you can get various insurance packages, and you will pay a total
price).  Each agent has preferences over the different possible ways
she can be served and we refer to this preference as her type.  The
seller is restricted by the {\em feasibility} constraint that at most
one agent can be served (e.g., only one car in the rental shop).  When
the agents' types are drawn independently from a known prior
distribution, the seller would like to design an auction to optimize
her objective, e.g., revenue, in expectation over this distribution,
subject to feasibility.  Importantly, in this abstract problem we have
not made any of the following standard assumptions on the agents'
preferences: quasi-linearity, risk-neutrality, or
single-dimensionality.

%
%

%
%
We assume that agents behave strategically and we will analyze an
auction's performance in {\em Bayes-Nash equilibrium}, i.e., where
each agent's strategy is a best response to the other agents'
strategies and the distribution over their preferences.  Without loss
of generality the revelation principle \citep{M81} allows for the
restriction of attention to {\em Bayesian incentive compatible} (BIC)
mechanisms, i.e., ones where the truthtelling strategy is a Bayes-Nash
equilibrium.

%
%
Any auction the seller proposes can be decomposed across the agents as
follows.  From an agent's perspective, the other agents are random
draws from the known distribution.  Therefore, the composition of
these random draws (as inputs), the bid of the agent, and the
mechanism induce an {\em interim allocation rule} which specifies the
probability the agent is served as a function of her bid.  BIC implies
that the agent is at least as happy to bid her type
as
any other bid.

%
%
Applying the same argument to each agent induces a profile of interim
allocation rules.  These interim allocation rules are jointly feasible
in the sense that there exists an auction that, for the prior
distribution, induces them.  As an example, suppose an agent's type is
high or low with probability $1/2$ each.  Consider two interim
allocation rules: rule (a) serves the agent with probability one when
her type is high and with probability zero otherwise, and rule (b)
serves the agent with probability $1/2$ regardless of her type.  It is
feasible for both agents to have rule (b) or for one agent to have
rule (a) and the other to have rule (b); on the other hand, it is
infeasible for both agents to have rule (a).  This last combination is
infeasible because with probability one quarter both agents have high
types but we cannot simultaneously serve both of them.  An important
question in the general theory of auctions is to decide when a profile
of interim allocation rules is feasible, and furthermore, when it is
feasible, to find an auction that implements it.

%
%
A structural characterization of the necessary and sufficient conditions
for the aforementioned {\em interim feasibility} is important for the
construction of optimal auctions as it effectively allows the auction
problem to be decomposed across agents.  If we can optimally serve a
single agent for a given interim allocation rule and we can check
feasibility of a profile of interim allocation rules, then we can
optimize over auctions.
Effectively, we can reduce the multi-agent
auction problem to a collection of single-agent auction
problems.

%
%

We now informally describe Border's \citeyearpar{B91} characterization
of interim feasibility.  A profile of interim allocation rules is
implementable if for any subspace of the agent types the expected
number of items served to agents in this subspace is at most the
probability that there is an agent with type in this subspace.
Returning to our infeasible example above, the probability that there
is an agent with a high type is $3/4$ while the expected number of
items served to agents with high types is one; Border's condition is
violated.

%
%
The straightforward formulation of interim feasibility via Border's
characterization has exponentially many constraints.  Nonetheless, it
can be simplified to a polynomial number of constraints in single-item
auctions with symmetric agents (where the agents' type space and
distribution are identical).  This simplification of the
characterization has lead to an analytically tractable theory of
auctions when agents have budgets \citep{LR96} or are risk averse
\citep{M84,MR84}.

%
%
\paragraph{Results}
Our main theorem is to show computationally tractable (i.e., in polynomial time in the total number of agents' types) methods for
each of the following problems. First, a given profile of interim allocation rules can be checked for interim feasibility.
Second, for any feasible profile of interim allocation rules, an auction (i.e., ex post allocation rule) that induces these
interim allocation rules can be constructed.  In particular, for problems where the seller can serve at most one agent, we show
that the exponentially-faceted polytope specified by the interim feasibility constraints is a projection of a
quadratically-faceted polytope in a higher dimension, and an ex post allocation rule implementing a feasible profile of interim
allocation rules is given immediately by the latter's preimage in the higher dimensional polytope. In particular, this implies
that optimal interim allocation rules can be computed by solving a quadratically sized linear/convex program. These results
combine to give a (computationally tractable) reduction from the multi-agent auction problem to a collection of single-agent
problems. Furthermore, our algorithmic procedure characterizes every single service auction as implementable by a simple token
passing game.

%
%
We also generalize the interim feasibility characterization and use it
to design optimal auctions for the cases where the seller faces a
$k$-unit feasibility constraint, i.e., at most $k$~agents can be
served, and more generally to matroid feasibility constraints.  Our
generalization of the feasibility characterization is based on a
simpler network-flow-based approach.  Although the number of
constraints in this characterization is exponential in the sizes of
type spaces, we observe that the constraints define a polymatroid,
whose vertices correspond to particularly simple auctions which can be
implemented by simple determinist allocation rules based on ranking.
Algorithms for submodular function minimization give rise to fast
separation oracles which, given a set of interim allocation rules,
detect a violated feasibility constraint whenever there is one;
expressing any point in the polymatroid as a convex combination of the
vertices allows us to implement any feasible interim allocation rule
as a distribution over the simple auctions.  These enable us again to
reduce the multi-agent problem to single-agent problems.


%
%
Auction theory is very poorly understood outside the standard
single-dimensional quasi-linear revenue maximization environment of
\citet{M81}.  The main consequence of this work is that even without
analytical understanding, optimal auctions can be computationally
solved for in environments that include non-quasi-linear utility
(e.g., budgets or risk aversion) and multi-dimensional preferences
(assuming that the corresponding single-agent problem can be solved).
Furthermore, unlike most work in auction theory with budgets or
risk-aversion, our framework permits the budgets or risk parameters to
be private to the agents.

\paragraph{Related Work}

\citet{M81} 
characterized Bayesian optimal auctions in environments with
quasi-linear risk-neutral single-dimensional agent
preferences.  \citet{BR89} reinterpreted Myerson's approach as reducing
the multi-agent auction problem to a related single-agent problem.
Our work generalize this reduction-based approach to single-item
multi-unit auction problems with general preferences.

An important aspect of our approach is that it can be applied to
general multi-dimensional agent preferences.  Multi-dimensional
preferences can arise as distinct values different configurations of
the good or service being auctioned, in specifying a private budget
and a private value, or in specifying preferences over risk.  We
briefly review related work for agent preferences with multiple
values, budgets, or risk parameters.

Multi-dimensional valuations are well known to be difficult.  For
example, \citet{RC98}, showed that, because {\em
bunching}\footnote{Bunching refers to the situation in which a group
of distinct types are treated the same way in by the mechanism.} can
not be ruled out easily, the optimal auctions for multi-dimensional
valuations are dramatically different from those for single
dimensional valuations.  Because of this, most results are for cases with special structure \citep[e.g.,][]{A96, W93, MM88} and
often, by using such structures, reduce the problems to
single-dimensional ones \citep[e.g.,][]{S80, R79, MS80}.  Our
framework does not need any such structure.


A number of papers consider optimal auctions for agents with budgets
\citep[see, e.g.,][]{PV08, CG95, M00}.  These papers rely on
budgets being public or the agents being symmetric; our technique
allows for a non-identical prior distribution and private budgets.
Mechanism design with risk averse agents was studied by
\citet{MR84} and \citet{Matthews83}.  Both works assume i.i.d.\@ prior
distributions and have additional assumptions on risk attitudes; our
reduction does not require these assumptions.

Our work is also related to a line of work on approximating the
Bayesian optimal mechanism.  These works tend to look for simple
mechanisms that give constant (e.g., two) approximations to the optimal
mechanism.  \citet{CHK07}, \citet{BCKW-10}, and \citet{CD11} consider
item pricing and lottery pricing for a single agent; the first two
give constant approximations the last gives a
$(1+\epsilon)$-approximation for any $\epsilon$.  These problems are
related to the single-agent problems we consider.
\citet{CHMS10} and \citet{BGGM10} extend these approaches to
multi-agent auction problems.  The point of view of reduction from
multi- to single-agent presented in this paper bears close
relationship to recent work by \citet{A11} who gives a reduction from
multi- to single-agent mechanism design that loses at most a constant
factor of the objective.  Our reductions, employing entirely different
techniques, give rise to optimal mechanisms instead of approximations
thereof.

Characterization of interim feasibility plays a vital role in this
work.  For single-item single-unit auctions, necessary and sufficient
conditions for interim feasibility were developed through a series of
works \citep{MR84, M84, B91, B07, M11}; this characterization has
proved useful for deriving properties of mechanisms, \citet{MV10}
being a recent example. \citet{B91} characterized symmetric interim
feasible auctions for single-item auctions with identically
distributed agent preferences.  His characterization is based on the
definition of ``hierarchical auctions.''  He observes that the space
of interim feasible mechanisms is given by a polytope, where vertices
of this polytope corresponding to hierarchical auctions, and interior
points corresponding to convex combinations of vertices.  \citet{M11}
generalize Border's approach and characterization to asymmetric
single-item auctions.  The characterization via hierarchical auctions
differs from our characterization via ordered subset auctions in that
hierarchical auctions allow for some types to be relatively unordered
with the semantics that these unordered types will be considered in a
random order; it is important to allow for this when solving for
symmetric auctions.  Of course convex combinations over hierarchical
auctions and ordered subset auctions provide the same generality.  Our
work generalizes the characterization from asymmetric single-unit
auctions to asymmetric multi-unit and matroid auctions.

Our main result provides computational foundations to the interim
feasibility characterizations discussed above.  We show that interim
feasibility can be checked, that interim feasible allocation rules can
be optimized over, and that corresponding ex post implementations can
be found.  Independently and contemporaneously \citet{CDW11} provided
similar computational foundations for the single-unit auction problem.
Their approach to the single-unit auction problem is most comparable
to our approach for the multi-unit and matroid auction problems where
the optimization problem is written as a convex program which can be
solved by the ellipsoid method; while these methods result in strongly
polynomial time algorithms they are not considered practical.  In
contrast, our single-unit approach, when the single-agent problems can
be solved by a linear program, gives a single linear program which can
be practically solved.

While our work gives computationally tractable interim feasibility
characterizations in ``service based'' environments like multi-unit
auctions and matroid auctions; \citet{CDW11} generalize the approach
to multi-item auctions with agents with additive preferences.  The
problem of designing an optimal auction for agents with
multi-dimensional additive preferences is considered one of the main
challenges for auction theory and their result, from a computational
perspective, solves this problem.

\paragraph{Organization.}
\label{sec:structure}

In Section~\ref{sec:prelim} we describe single- and multi-agent
mechanism design problems.  In Section~\ref{sec:single-agent} we give
algorithms for solving two kinds of single-agent problems: multi-item
unit-demand preferences and private-value private-budget preferences.
In Section~\ref{sec:alg}, we give a high-level description of the
multi- to single-agent reduction which allows for efficiently compute
optimal mechanisms for many service based environments.  The key step
therein, an efficient algorithm that implements any jointly feasible
set of interim allocation rules, is presented in
Section~\ref{sec:inter}.  This section is divided into three parts
which address single-unit, multi-unit, and matroid feasibility
constraints, respectively.  Conclusions and extensions are discussed
in Section~\ref{sec:conc}.

\section{Preliminaries}
\label{sec:prelim}

\paragraph{Single-agent Mechanisms}


We consider the provisioning of an abstract service.  This service may
be parameterized by an {\em attribute}, e.g., quality of service, and
may be accompanied by a required payment.  We denote the outcome
obtained by an agent as $\outcome \in \outcomespace$.  We view this
outcome as giving an indicator for whether or not an agent is served and
as describing attributes of the service such as quality of service and
monetary payments.  Let $\ALLOC(\outcome) \in \{0,1\}$ be an indicator
for whether the agent is served or not; let $\PAYMENT(\outcome)
\in \reals$ denote any payment the agent is required to make.  In a
randomized environment (e.g., randomness from a randomized mechanism or Bayesian
environment) the outcome an agent receives is a random variable from a
distribution over $\outcomespace$.  The space of all such
distributions is denoted $\outcomedistspace$.

The agent has a type $\type$ from a finite type space $\typespace$.
This type is drawn from distribution $\dens \in \distover{\typespace}$
and we equivalently denote by $\dens$ the probability mass function.
I.e., for every $\type \in \typespace$, $\dens(\type)$ is the
probability that the type is~$\type$.  The utility function $\util
\,:\, \typespace \times \outcomespace \to \reals$ maps the agent's
type and the outcome to real valued utility.  The agent is a von
Neumann--Morgenstern expected utility maximizer and we extend $\util$
to $\distoutcomespace$ linearly, i.e., for $\outcome
\in \distoutcomespace$, $\util(\type, \outcome)$ is the expectation of $\util$ where the outcome is drawn according
to~$\outcome$.
We do not require the usual assumption of quasi-linearity.


A single-agent mechanism, without loss of generality by the revelation
principle, is just an {\em outcome rule}, a mapping from
the agent's type to a distribution over outcomes.  We denote an {\em
  outcome rule} by $\toutcome \,:\, \typespace
\to \outcomedistspace$.  We say that an outcome rule $\toutcome$ is
{\em incentive compatible} (IC) and {\em individually rational} (IR) if for
all $\type,\type' \in \typespace$, respectively,
\begin{align}
\tag{IC}\label{eq:IC}
\util(\type,\toutcome(\type)) &\geq \util(\type,\toutcome(\type')),\\
\tag{IR}\label{eq:IR}
\util(\type,\toutcome(\type)) &\geq 0.
\end{align}

We refer to restriction of the outcome rule to the indicator for
service as the {\em allocation rule}.  As the allocation to
each agent is a binary random variable, distributions over allocations
are fully described by their expected value.  Therefore the allocation
rule $\talloc \,:\, \typespace \to [0,1]$ for a given outcome rule
$\outcome$ is $\talloc(\type) = \expect{\ALLOC(\toutcome(\type))}$.

We give two examples to illustrate the abstract model described above.
The first example is the standard quasi-linear risk-neutral preference
which is prevalent in auction theory.  Here the agent's type space is
$\typespace \subset \posreals$ where $\type \in \typespace$ represents the
agent's valuation for the item.  The outcome space is $\outcomespace =
\{0, 1\} \times \posreals$ where an outcome $\outcome$ in this space
indicates whether or not the item is sold to the agent, by
$\ALLOC(\outcome)$, and at what price, by $\PAYMENT(\outcome)$.  The
agent's quasi-linear utility function is $\util(\type, \outcome) =
\type \cdot \ALLOC(\outcome) - \PAYMENT(\outcome)$.
%
%
The second example is that of an $\numservice$-item unit-demand (also
quasi-linear and risk-neutral) preference.  Here the type space is
$\typespace \subset \posreals^\numservice$ and a type $\type
\in \typespace$ indicates the agent's valuation for each of the items
when the agent's value for no service is normalized to zero.  An
outcome space is $\outcomespace = \{0,\ldots,\numservice\} \times
\posreals$.  The first coordinate of $\outcome$ specifies which item
the agent receives or none and $\ALLOC(\outcome) = 1$ if it is
non-zero; the second coordinate of $\outcome$ specifies the required
payment $\PAYMENT(\outcome)$.  The agent's utility for $\outcome$ is
the value the agent attains for the item received less her payment.
Beyond these two examples, our framework can easily incorporate more
general agent preferences exhibiting, e.g., risk aversion or a budget
limit.

Consider the following single-agent mechanism design problem.  A
feasibility constraint is given by an upper bound $\talloc(\type)$ on
the probability that the agent is served as a function of her type
$\type$; the distribution on types in $\typespace$ is given by
$\dens$.  The single-agent problem is to find the outcome rule
$\toutcome^*$ that satisfies the allocation constraint of $\talloc$
and maximizes the performance, e.g., revenue.  This problem is
described by the following program:
\begin{align*}
\max_{\toutcome}:  & \quad \expect[\type \sim \dens, \toutcome(\type)]{\PAYMENT(\toutcome(\type))}
\tag{SP}\label{eq:rule-rev-def} \\
\hbox{s.t. } & \quad \expect[\toutcome(\type)]{\ALLOC(\toutcome(\type))} \leq \talloc(\type), \qquad \forall \type \in \typespace \\
 & \quad \toutcome \text{ is IC and IR}.
\end{align*}
We denote the outcome rule $\toutcome^*$ that optimizes this program
by $\RULE(\talloc)$ and its revenue by $\REV(\talloc) = \expect[\type
  \sim \dens, \toutcome^*(\type)]{\PAYMENT(\toutcome^*(\type))}$.  We note
that, although this paper focuses on revenue maximization, the same
techniques presented can be applied to maximize (or minimize) general
separable objectives such as social welfare.

\paragraph{Multi-agent Mechanisms}

There are $n$ independent agents.  Agents need not be identical,
i.e., agent $\agent$'s type space is~$\typespace\agind$, the
probability mass function for her type is~$\dens\agind$,
her outcome space is $\outcomespace\agind$, and her utility function
is $\util\agind$. The profile of agent types is denoted by $\types =
(\type\agind[1],\ldots,\type\agind[n]) \in \typespace\agind[1] \times
\cdots \times \typespace\agind[n] = \typespaces$, the joint
distribution on types is $\denss \in \distover{\typespace_1}\times \cdots \times \distover{\typespace_n}$, a vector of outcomes is
$(\outcome\agind[1], \cdots, \outcome\agind[n]) \in \outcomespaces$, and an allocation is
$(\alloc\agind[1],\ldots,\alloc\agind[n]) \in \{0,1\}^n$.  The
mechanism has an inter-agent feasibility constraint that permits
serving at most $k$~agents, i.e., $\sum_\agent \alloc\agind \leq
k$.\footnote{Furthremore, in Section~\ref{sec:border_matroid}, we
  review the theory of {\em matroids} and extend our basic results
  environments with feasibility constraint derived from a matroid set
  system.}  A mechanism that obeys this constraint is {\em
  feasible}. The mechanism has no inter-agent constraint on attributes
or payments.

A mechanism maps type profiles to a (distribution over) outcome
vectors via an {\em ex post outcome rule}, denoted $\epoutcomes \,:\,
\typespaces \to \outcomedistspaces$ where $\epoutcome\agind(\types)$
is the outcome obtained by agent $\agent$. We will similarly define
$\epallocs \,:\, \typespaces \to [0,1]^n$ as the {\em ex post
  allocation rule} (where $[0,1] \equiv \distover{\{0,1\}}$).  The ex
post allocation rule~$\epallocs$ and the probability mass
function~$\denss$ on types induce {\em interim} outcome and allocation
rules.  For agent~$\agent$ with type~$\type\agind$ and $\types \sim
\distribution[\types]{\types \given \type\agind}$ the interim outcome
and allocation rules are $\intoutcome\agind(\type\agind) =
\distribution[\types]{\epoutcome\agind(\types) \given \type\agind }$
and $\intalloc\agind(\type\agind) =
\distribution[\types]{\epalloc\agind(\types) \given \type\agind}
\equiv \expect[\types]{\epalloc\agind(\types) \given
  \type\agind}$.\footnote{We use notation $\distribution{X \given E}$
  to denote the distribution of random variable $X$ conditioned on the
  event $E$.}  A profile of interim allocation rules is feasible if it
is derived from an ex post allocation rule as described above; the set
of all feasible interim allocation rules is denoted by $\InAllocSpace$.
A mechanism is Bayesian incentive compatible and interim individually
rational if equations~\eqref{eq:IC} and~\eqref{eq:IR}, respectively,
hold for all $\agent$ and all~$\type\agind$.

Consider again the examples described previously of quasi-linear
single-dimensional and unit-demand preferences.  For the
single-dimensional example, the multi-agent mechanism design problem
is the standard single-item $k$-unit auction problem.  For the
unit-demand example, the multi-agent mechanism design problem is an
{\em attribute auction}.  In this problem there are $k$-units
available and each unit can be configured in one of $\numservice$
ways.  Importantly, the designer's feasibility constraint restricts
the number of units sold to be $k$ but places no restrictions on how
the units can be configured.  E.g., a restaurant has $k$ tables but
each diner can order any of the $m$ entrees on the menu.


A reduction from multi-agent mechanism design to single-agent
mechanism design as we have described above would assume that for any
types pace $\typespace\agind$, any probability mass
function~$\dens\agind$, and interim allocation rule $\intalloc\agind$,
the optimal outcome rule $\RULE(\intalloc\agind)$ and its performance
$\REV(\intalloc\agind)$ can be found efficiently (see
Section~\ref{sec:single-agent} for examples).  The goal then is to
construct an optimal multi-agent auction from these single-agent
mechanisms.  Our approach to such a reduction is
as follows.
\begin{enumerate}
\item \label{step:optimize}
Optimize, over all feasible profiles of interim allocation rules
$\intallocs = (\intalloc\agind[1],\ldots,\intalloc\agind[n]) \in \InAllocSpace$, the
sum of performances of the allocation
rules $\sum_\agent \REV(\intalloc\agind)$.
\item \label{step:expostize} Implement the profile of interim outcome
  rules $\intoutcomes$ given by $\intoutcome\agind =
  \RULE(\intalloc\agind)$ with a feasible ex post outcome rule $\epoutcomes$.
\end{enumerate}

Two issues should be noted.  First, Step~\ref{step:expostize} requires
an argument that the existence of a feasible ex post outcome rule for
a given profile of interim allocation rules implies the existence of
one that combines the optimal interim outcome rules
from~$\RULE(\cdot)$.  We address this issue in Section~\ref{sec:alg}.
Second, Step~\ref{step:optimize} requires that we optimize over
jointly feasible interim allocation rules, and after solving for
$\intallocs$, its implementation by an ex post allocation rule is
needed to guide Step~\ref{step:expostize}.  We address this issue in
Section~\ref{sec:ssa}.  For single-unit (i.e., $k=1$) auctions a
characterization of the necessary and sufficient condition for interim
feasibility was provided by Kim Border.




\begin{theorem}[\citealp{B91}]
\label{thm:border}%
In a single-item auction environment, interim allocation rules
$\intallocs$ are feasible (i.e., $\intallocs \in \InAllocSpace$) if
and only if the following holds:
\begin{align}
\forall S\agind[1] \subseteq \typespace\agind[1],
\cdots, \forall S\agind[n] \subseteq \typespace\agind[n]: &\qquad
\sum_{\agent=1}^n \expect{\intalloc\agind(t_i) \given \type\agind \in
    S\agind}\cdot \prob{\type\agind \in S\agind}
        \le \prob[\types \sim \denss]{\exists i \in[n]: \type\agind \in S\agind
    } \tag{MRMB} \label{eq:MRM}
\end{align}
\end{theorem}

\section{The Single-agent Problem}
\label{sec:single-agent}

Given an allocation rule $\talloc(\cdot)$ as a constraint the
single-agent problem is to find the (possibly randomized) outcome rule
$\toutcome(\cdot)$ that allocates no more frequently that
$\talloc(\cdot)$, i.e., $\forall \type \in \typespace$,
$\expect[\toutcome(\type)]{\ALLOC(\toutcome(\type))} \leq
\talloc(\type)$, with the maximum expected performance.  Recall that the optimal
such outcome rule is denoted $\RULE(\talloc)$ and its performance
(e.g., revenue) is denoted $\REV(\talloc)$.  We first observe that $\REV(\cdot)$ is concave.

\begin{proposition}
\label{claim:concave}
$\REV(\cdot)$ is a concave function in $\Inter \alloc$.
\end{proposition}
\begin{proof}
Consider any two allocation rules $\intalloc$ and $\intalloc'$,
and any $\alpha \in [0,1]$.  Define $\intalloc''$ to be $\alpha
\intalloc + (1-\alpha) \intalloc'$.  We will show that $\alpha
\REV(\intalloc) + (1 - \alpha) \REV(\intalloc') \leq
\REV(\intalloc'')$, which proves the claim.  To see this, let
$\toutcome$ and $\toutcome'$ be $\RULE(\intalloc)$ and
$\RULE(\intalloc')$, respectively.
Define $\toutcome''$ to be the outcome rule that runs $\toutcome$ with
probability~$\alpha$, and $\toutcome'$ with probability $1-\alpha$.
The incentive compatibility of outcome rules $\toutcome$ and
$\toutcome'$ imply the incentive compatibility of $\toutcome''$, since
for any $\type, \type' \in \typespace$,
\begin{align*}
\expect{\util(\type, \toutcome''(\type))} 
  &= \alpha \expect{\util(\type, \toutcome(\type))} + (1 - \alpha)
\expect{\util(\type, \toutcome'(\type))} \\
&\geq  \alpha \expect{\util(\type, \toutcome(\type'))} + (1 - \alpha)
\expect{\util(\type, \toutcome'(\type'))}\\
& = \expect {\util(\type,\toutcome''(\type'))}.
\end{align*}
Also, $\toutcome''$ is feasible as $\expect{\ALLOC(\toutcome''(\type))} = \alpha
\expect{\ALLOC(\toutcome(\type))} + (1-\alpha)
\expect{\ALLOC(\toutcome'(\type))} \leq \intalloc''(\type)$ for all
$\type \in \typespace$. As a result, $\REV(\intalloc'')$ is at least
the revenue of $\toutcome''$, which is in turn equal to $\alpha
\REV(\intalloc) + (1-\alpha) \REV(\intalloc')$.
\end{proof}

We now give two examples for which the single-agent problem is
computationally tractable.  Both of these example are
multi-dimensional.  The first example is that of a standard multi-item
unit-demand preferences.  The second example that of a single-item
with a private budget.  For both of these problems the single-agent
problem can be expressed as a linear program with size polynomial in
the cardinality of the agent's type space.

\subsection{Quasi-linear Unit-demand Preferences}

There are $\numservice$ items available.  There is a finite type space
$\typespace \subset \posreals^\numservice$; the outcome space
$\outcomespace$ is the direct product between an assignment to the
agent of one of the $\numservice$ items, or none, and a required
payment.  $\outcomedistspace$ is the cross product of a probability
distribution over which item the agent receives and a probability
distribution over payments.  Without loss of generality for a
quasi-linear agent such a randomized outcome can be represented as
$\outcome =
(\mdalloc{\outcome}{1},\ldots,\mdalloc{\outcome}{\numservice},\mdprice{\outcome})$
where for $j \in [\numservice]$, $\mdalloc{\outcome}{j}$ is the
probability that the agent receives item $j$ and $\mdprice{\outcome}$
is the agent's required payment.

A single-agent mechanism assigns to each type an outcome as described
above.  An outcome rule specifies an outcome for any type $\type$ of
the agent as $\toutcome(\type) =
(\mdalloc{\toutcome}{1}(\type),\ldots,\mdalloc{\toutcome}{\numservice}(\type),\mdprice{\toutcome}(\type))$.
This gives $\numservice+1$ non-negative real valued variables for each of
$|\typespace|$ types.  The following linear program, which is a simple
adaptation of one from \citet{BCKW-10} to include the feasibility
constraint given by $\talloc$, solves for the optimal single-agent
mechanism:
\begin{align*}
\max: & \quad \sum\nolimits_{\type \in \typespace} \dens(\type) \mdprice{\outcome}(\type) & \\
\hbox{s.t. }&\quad\!\sum\nolimits_j \mdalloc{\toutcome}{j}(\type) \leq \intalloc(\type) & \forall \type \in \typespace \\
            & \quad \sum\nolimits_j \mdval{\type}{j} \mdalloc{\toutcome}{j}(\type) - \mdprice{\toutcome}(\type) \geq \sum\nolimits_j \mdval{\type}{j} \mdalloc{\toutcome}{j}(\type') - \mdprice{\toutcome}(\type') & \forall \type,\type' \in \typespace\\
            & \quad \sum\nolimits_j \mdval{\type}{j} \mdalloc{\toutcome}{j}(\type) - \mdprice{\toutcome}(\type) \geq 0 & \forall \type \in \typespace.
\end{align*}
The optimal outcome rule from this program is $\toutcome^* = \RULE(\talloc)$ and its performance is $\REV(\talloc) = \expect[\type \sim \dens]{\mdprice{\toutcome^*}(\type)}$.

\begin{proposition}
The single-agent $\numservice$-item unit-demand problem can be solved
in polynomial time in $\numservice$ and $|\typespace|$.
\end{proposition}

\subsection{Private budget preferences.}

There is a single item available.  The agent has a private value for
this item and a private budget, i.e., $\typespace \subset
\posreals^2$; we will denote by $\bugval{\type}$ and $\bugbug{\type}$
this value and budget respectively.  The outcome space is
$\outcomespace = \{0,1\} \times \reals$ where for $\outcome
\in \outcomespace$ the first coordinate $\bugalloc{\outcome}$ denotes
whether the agent receives the item or not and the second coordinate
$\bugprice{\outcome}$ denotes her payment.  The agent's utility
is 
\[\util(\type,\outcome) =
\begin{cases}
\bugval{\type} \bugalloc{\outcome} - \bugprice{\outcome} & \text{if $\bugprice{\outcome} \leq \bugbug{\type}$, and }\\
-\infty & \text{otherwise.} 
\end{cases}
\] 
\autoref{claim:private-budget} below implies that when optimizing over
distributions on outcomes we can restrict attention to $[0,1] \times
[0,1] \times \posreals \subset \outcomedistspace$ where the first
coordinate denotes the probability that the agent receives the item,
the second coordinate denotes the probability that the agent makes a
non-zero payment, and the third coordinate denotes the non-zero
payment made.

\begin{claim}
\label{claim:private-budget}
Any incentive compatible and individually rational outcome rule can
be converted into an outcome rule above with the same expected
revenue.
\end{claim}

As a sketch of the argument to show this claim, note that if an agent
with type $\type$ receives randomized outcome $\outcome$ she is just
as happy to receive the item with the same probability and pay her
budget with probability equal to her previous expected payment divided
by her budget.  Such a payment is budget feasible and has the same
expectation as before.  Furthermore, this transformation only increases
the maximum payment that any agent makes which means that the relevant
incentive compatibility constraints are only fewer.  Importantly, the
only incentive constraints necessary are ones that prevent types with
higher budgets from reporting types with lower budgets.

A single-agent mechanism assigns to each type an outcome as described
above.  We denote the distribution over outcomes for $\type$ by
$\toutcome(\type) =
(\bugalloc{\outcome}(\type),\bugprob{\outcome}(\type),\bugbug{\type})$
where only the first two coordinates are free variables.  This gives
two non-negative real valued variables for each of $|\typespace|$
types.  The following linear program solves for the optimal
single-agent mechanism:
\begin{align*}
\max: & \quad \sum\nolimits_{\type \in \typespace} \dens(\type) \bugbug{\type} \bugprob{\outcome}(\type) & \\
\hbox{s.t. }&\quad \bugalloc{\outcome}(\type) \leq \intalloc(\type) & \forall \type \in \typespace \\
            & \quad \bugval{\type} \bugalloc{\outcome}(\type) - \bugbug{\type}\bugprob{\outcome}(\type) \geq \bugval{\type} \bugalloc{\outcome}(\type') - \bugbug{\type'}\bugprob{\outcome}(\type') & \forall \type,\type' \in \typespace \text{ with $\bugbug{\type'} \leq \bugbug{\type}$}\\
            & \quad \bugval{\type} \bugalloc{\outcome}(\type) - \bugbug{\type}\bugprob{\outcome}(\type) \geq 0 & \forall \type \in \typespace\\
            & \quad \bugprob{\outcome}(\type) \leq 1 & \forall \type \in \typespace.
\end{align*}
The optimal outcome rule from this program is $\toutcome^* = \RULE(\talloc)$ and its performance is $\REV(\talloc) = \expect[\type \sim \dens]{\bugbug{\type}\bugprob{\toutcome^*}(\type)}$.

\begin{proposition}
The single-agent private budget problem can be solved in polynomial
time in $|\typespace|$.
\end{proposition}

\section{Multi- to Single-agent Reductions}
\label{sec:alg}%

An ex post allocation rule $\epallocs$ takes as its input a profile of
types $\Types=(\Type_1,\ldots, \Type_n)$ of the agents, and indicates
by $\epalloc_i(\types)$ a set of at most $k$ winners.  Agent $i$'s type
$\type_i \in \typespace_i$ is drawn independently at random from
distribution $\dens_i \in \distover{\typespace_i}$.  An ex post
allocation rule implements an interim allocation rule $\InAlloc_i
: \TypeSpace_i \to [0,1]$, for agent~$i$, if the probability of winning
for agent $i$ conditioned on her type $\Type_i \in \TypeSpace_i$ is
exactly $\InAlloc_i(\Type_i)$, where the probability is taken over the
random types other agents and the random choices of the allocation
rule.  A profile of interim allocation rules
$\InAllocs=(\InAlloc_1,\ldots, \InAlloc_n)$ is feasible if and only if
it can be implemented by some ex post allocation rule. $\InAllocSpace$
denotes the space of all feasible profiles of interim allocation
rules.

The optimal performance (e.g., revenue) of the single-agent problem
with allocation constraint given by $\talloc$ is denoted
$\REV(\talloc)$.  The outcome rule corresponding to this optimal
revenue is $\RULE(\talloc)$.  Given any feasible interim allocation
rule $\intallocs \in \InAllocSpace$ we would like to construct an
auction with revenue $\sum\agind \REV(\talloc\agind)$.  We need to be
careful because $\RULE(\intalloc\agind)$, by definition, is only
required to have allocation rules \emph{upper bounded}
by~$\intalloc\agind$ (see \eqref{eq:rule-rev-def} in
Section~\ref{sec:prelim}), while the ex post allocation
rule~$\epalloc\agind$ implements $\intalloc\agind$ exactly, and hence
we may need to scale down~$\epalloc\agind$ accordingly.  This is
defined formally as follows.

\begin{definition}
An optimal auction $\epoutcomes^*$ for feasible interim allocation rule $\intallocs$ (with corresponding ex post allocation rule $\epallocs$) is defined as follows on $\types$.  For agent $\agent$:
\begin{enumerate}
\item Let $\intoutcome^*\agind = \RULE(\intalloc\agind)$ be the optimal outcome rule for allocation constraint $\intalloc\agind$.
\item Let $\intalloc^*\agind = \expect{\ALLOC(\intoutcome^*\agind)}$
  be the allocation rule corresponding to outcome rule $\intoutcome^*\agind$.
\item If $\epalloc\agind(\types) = 1$, output
\[
\epoutcome^*\agind(\types) \sim
\begin{cases}
 \distribution{\intoutcome^*\agind(\type\agind) \given \ALLOC(\intoutcome^*\agind(\type\agind)) = 1}  & \text{w.p.~$\intalloc^*\agind(\type\agind)/\intalloc\agind(\type\agind)$, and}\\
\distribution{\intoutcome^*\agind(\type\agind) \given \ALLOC(\intoutcome^*\agind(\type\agind)) = 0} &\text{otherwise.}
\end{cases}
\]
\item Otherwise (when $\epalloc\agind(\types) = 0$), output
$\epoutcome^*\agind(\types) \sim \distribution{\intoutcome^*\agind(\type\agind) \given \ALLOC(\intoutcome^*\agind(\type\agind)) = 0}.$
\end{enumerate}
\end{definition}

\begin{proposition}
\label{prop:implementation}
For any feasible interim allocation rule $\intallocs
\in \InAllocSpace$, the optimal auction for this rule has expected
revenue $\sum_i \REV(\intalloc\agind)$.
\end{proposition}

\begin{proof}
The ex post outcome rule $\epoutcomes^*$ of the auction, by
construction, induces interim outcome rule $\intoutcomes^*$ for which
the revenue is as desired.
\end{proof}

The optimal multi-agent auction is the solution to optimizing the
cumulative revenue of individual single-agent problems subject to the
joint interim feasibility constraint given by $\intallocs
\in \InAllocSpace$.

\begin{proposition}
\label{claim:outcome-impl}
The optimal revenue is given by the convex program
\begin{align}
\max_{\intallocs \in \InAllocSpace}: \sum\nolimits_{\agent} \REV\agind(\intalloc\agind). \tag{CP} \label{eq:CP}
\end{align}
\end{proposition}

\begin{proof}
This is a convex program as $\REV(\cdot)$ is concave and
$\InAllocSpace$ is convex (convex combinations of feasible interim
allocation rules are feasible).  By \autoref{prop:implementation} this
revenue is attainable; therefore, it is optimal.
\end{proof}

\section{Optimization and Implementation of Interim Allocation Rules}
\label{sec:inter}%

In this section we address the computational issues pertaining to (i) solving optimization problems over the space of feasible
interim allocation rules, and (ii) ex post implementation of such a feasible interim allocation rule.  We present computationally
tractable methods for both problems.

\paragraph{Normalized interim allocation rules.}
It will be useful to ``flatten'' the interim allocation rule $\InAllocs$ for which $\InAllocs_i (\type_i)$ denotes the
probability that $i$ with type $\Type_i$ is served (randomizing over the mechanism and the draws of other agent types); we do so
as follows. Without loss of generality, we assume that the type spaces of different agents are disjoint.\footnote{This can be
achieved by labeling all types of
  each agent with the name of that agent, i.e., for each $i\in [n]$ we
  can replace $\TypeSpace_i$ with $\TypeSpace'_i=\{(i,t) | t
  \in \TypeSpace_i\}$ so that $\TypeSpace'_1,\cdots, \TypeSpace'_n$
  are disjoint.}  Denoting the set of all types by $\UnionTypeSpace =
\bigcup_i \TypeSpace_i$, the interim allocation rule can be flattened as a vector in $[0,1]^{\UnionTypeSpace}$.

\begin{definition}
The {\em normalized interim allocation rule} $\PreAlloc \in
[0,1]^{\UnionTypeSpace}$ corresponding to interim allocation rule
$\intallocs$ under distribution $\DistProbs$ is defined as
\begin{align*}
    \PreAlloc(\Type_i) &= \intalloc_i(\Type_i) \DistProb_i(\Type_i)&
    &\forall \Type_i \in \UnionTypeSpace
\end{align*}
\end{definition}
For the rest of this section, we refer to interim allocation rules via $\PreAlloc$ instead of $\InAllocs$.  Note that there is a
one-to-one correspondence between $\PreAlloc$ and $\InAllocs$ as specified by the above linear equation; so any linear of convex
optimization problem involving $\InAllocs$ can be written in terms of $\PreAlloc$ without affecting its linearity or convexity.
As $\InAllocSpace$ denotes the space of feasible interim allocation rules $\InAllocs$, we will use $\PreAllocSpace$ to denote the
space of feasible normalized interim allocation rules.



In the remainder of this section we characterize interim feasibility
and show that normalized interim allocation rules can be optimized
over and implemented in polynomial time.

\subsection{Single Unit Feasibility Constraints}
\label{sec:border_1}%

In this section, we consider environments where at most one agent can be allocated to. For such environments, we characterize
interim feasibility as implementability via a particular, simple {\em stochastic sequential allocation} mechanism. Importantly,
the parameters of this mechanism are easy to optimize efficiently.

A stochastic sequential allocation mechanism is parameterized by a stochastic transition table.  Such a table specifies the
probability by which an agent with a given type can steal a token from a preceding agent with a given type.  For simplicity in
describing the process we will assume the token starts under the possession of a ``dummy agent'' indexed by 0; the agents are
then considered in the arbitrary order from 1 to $n$; and the agent with the token at the end of the process is the one that is
allocated (or none are allocated if the dummy agent retains the token).

\begin{definition}[stochastic sequential allocation mechanism]
\label{def:ssa}%
Parameterized by a stochastic transition table $\SeqTable$, the {\em stochastic sequential allocation mechanism (SSA)} computes
the allocations for a type profile $\types \in \TypeSpaces$ as follows:
\begin{enumerate}
\item Give the token to the dummy agent $0$ with dummy type $\type_0$.
\item For each agent $i$: (in order of $1$ to $n$)

If agent $i'$ has the token, transfer the token to
agent $i$ with probability $\pi(\type_{i'},\type_i)$.

\item Allocate to the agent who has the token (or none if the
  dummy agent has it).
\end{enumerate}
\end{definition}

First, we present a dynamic program, in the form of a collection of linear equations, for calculating the interim allocation rule
implemented by SSA for a given $\SeqTable$. Let $\SeqAlloc(\Type_{i'},i)$ denote the ex-ante probability of the event that
agent~$i'$ has type $\Type_{i'}$ and is holding the token at the end of iteration $i$. Let $\SeqTrans(\Type_{i'}, \Type_{i})$
denote the ex-ante probability in iteration $i$ of SSA that agent~$i$ has type~$\Type_i$ and takes the token from agent~$i'$ who
has type $\Type_{i'}$.

The following additional notation will be useful in this section.  For any subset of agents $\AgentSubset \subseteq \Agents =
\{1,\ldots, n\}$, we define $\TypeSpace_\AgentSubset = \bigcup_{i \in
  \AgentSubset} \TypeSpace_i$ (Recall that without loss of generality
agent type spaces are assumed to be disjoint.).  The shorthand notation $\Type_i \in \TypeSubset$ for $\TypeSubset \subseteq
\UnionTypeSpace$ will be used to quantify over all types in $\TypeSubset$ and their corresponding agents (i.e., $\forall \Type_i
\in \TypeSubset$ is equivalent to $\forall i \in \Agents, \forall \Type_i \in \TypeSubset\cap \TypeSpace_i$).

The normalized interim allocation rule $\PreAlloc$ resulting from the
SSA is exactly given by the dynamic program specified by the following
linear equations.

\newcounter{SeqCounter}
\newcommand{\SeqTag}{{\protect\stepcounter{SeqCounter}$\SeqAllocSpace$.\arabic{SeqCounter}}}

\begin{align}
    \SeqAlloc(\Type_0,0) & = 1, &
        & \tag{\SeqTag} \label{eq:seq_alloc0} \\
    \SeqAlloc(\Type_i,i) &= \sum_{\Type_{i'} \in \TypeSpace_{\Range{0}{i-1}}} \SeqTrans(\Type_{i'},\Type_i), &
        & \forall \Type_i \in \TypeSpace_{\Range{1}{n}}  \tag{\SeqTag} \label{eq:seq_alloc_self} \\
    \SeqAlloc(\Type_{i'},i) &= \SeqAlloc(\Type_{i'}, i-1) - \sum_{\Type_{i} \in \TypeSpace_{i}} \SeqTrans(\Type_{i'}, \Type_{i}), &
        & \forall i \in \Range{1}{n}, \forall \Type_{i'} \in \TypeSpace_{\Range{0}{i-1}} \tag{\SeqTag} \label{eq:seq_alloc}  \\
    \SeqTrans(\Type_{i'}, \Type_{i}) &= \SeqAlloc(\Type_{i'}, i-1) \SeqTable(\Type_{i'}, \Type_{i}) \DistProb_{i}(\Type_{i}), &
        & \forall \Type_{i} \in \TypeSpace_{\Range{1}{n}}, \forall \Type_{i'} \in \TypeSpace_{\Range{0}{i-1}} \tag{$\pi$}
        \label{eq:seq_pi} \\
    \PreAlloc(\Type_i) &=\SeqAlloc(\Type_i,n), &
        & \forall \Type_i \in \TypeSpace_{\Range{1}{n}} \notag 
\intertext{Note that $\SeqTable$ is the only adjustable parameter in the SSA algorithm, so by relaxing the equation
\eqref{eq:seq_pi} and replacing it with the following inequality we can specify all possible dynamics of the SSA algorithm. }
    0 \le \SeqTrans(\Type_{i'}, \Type_{i}) &\le \SeqAlloc(\Type_{i'}, i-1) \DistProb_{i}(\Type_{i}), &
        & \forall \Type_{i} \in \TypeSpace_{\Range{1}{n}}, \forall \Type_{i'} \in \TypeSpace_{\Range{0}{i-1}} \tag{\SeqTag} \label{eq:seq_last}
\end{align}

Let $\SeqAllocSpace$ denote the convex polytope captured by the \arabic{SeqCounter}~sets of linear constraints
\eqref{eq:seq_alloc0} through \eqref{eq:seq_last} above, i.e., $(\SeqAllocs, \SeqTranses) \in \SeqAllocSpace$ if{f} $\SeqAllocs$
and $\SeqTranses$ satisfy the aforementioned constraints. Note that every $(\SeqAllocs, \SeqTranses) \in \SeqAllocSpace$
corresponds to some stochastic transition table $\SeqTable$ by solving equation~\eqref{eq:seq_pi} for
$\SeqTable(\type_i,\type_{i'})$. We show that $\SeqAllocSpace$ captures all feasible normalized interim allocation rules, i.e.,
the projection of $\SeqAllocSpace$ on $\PreAlloc(\cdot)=\SeqAllocs(\cdot, n)$ is exactly $\PreAllocSpace$, as formally stated by
the following theorem.

\begin{theorem}
\label{thm:seq_alloc}%
A normalized interim allocation rule $\PreAlloc$ is feasible if and only if it can be implemented by the SSA algorithm for some
choice of stochastic transition table $\SeqTable$. In other words, $\PreAlloc \in \PreAllocSpace$ if{f} there exists
$(\SeqAllocs, \SeqTranses) \in \SeqAllocSpace$ such that $\PreAlloc(\Type_i)=\SeqAllocs(\Type_i, n)$ for all $\Type_i \in
\UnionTypeSpace$.
\end{theorem}

\begin{corollary}
Given a blackbox for each agent~$i$ that solves for the optimal expected revenue $\Rev_i(\InAlloc_i)$ for any feasible interim
allocation rule $\InAllocs$, the optimal interim allocation rule can be computed by the following convex program which is of
quadratic size in the total number of types.
\begin{alignat*}{3}
    \text{maximize}     &\qquad& & \sum_{i=1}^n \Rev_i(\InAlloc_i) & \qquad & \\
    \text{subject to}   && & \SeqAllocs(\Type_i, n) = \PreAlloc(\Type_i) = \InAlloc_i(\Type_i) \DistProb_i(\Type_i),&
                            &\forall \Type_i \in \UnionTypeSpace\\
                        && & (\SeqAllocs, \SeqTranses) \in \SeqAllocSpace.
\end{alignat*}
Furthermore, given an optimal assignment for this program, the computed interim allocation rule can be implemented by SSA using
the the stochastic transition table defined by:\footnote{If the denominator is
  zero, i.e., $\SeqAlloc(\Type_i, i'-1)=0$, we can set
  $\SeqTable(\Type_i, \Type_{i'})$ to an arbitrary value in $[0,1]$.}
\begin{align*}
    \SeqTable(\Type_{i'}, \Type_{i}) &= \frac{\SeqTrans(\Type_{i'}, \Type_{i})}{\SeqAlloc(\Type_{i'}, i-1) \DistProb_{i}(\Type_{i})}, &
        & \forall \Type_{i} \in \TypeSpace_{\Range{1}{n}}, \forall \Type_{i'} \in \TypeSpace_{\Range{0}{i-1}}.
\end{align*}
\end{corollary}

Next, we present a few definitions and lemmas that are used in the
proof of \autoref{thm:seq_alloc}.  Two transition tables $\SeqTable$
and $\SeqTable'$ are considered \emph{equivalent} if their induced
normalized interim allocation rules for SSA are equal.  Type $\Type_i$
is called \emph{degenerate} for $\SeqTable$ if in the execution of SSA
the token is sometimes passed to type $\Type_i$ but it is always taken
away from $\Type_i$ later, i.e., if $\SeqAllocs(\Type_i,i) > 0$ but
$\SeqAllocs(\Type_i,n) = 0$.  The stochastic transition table
$\SeqTable$ is degenerate if there is a degenerate type. For
$\SeqTable$, type $\Type_i$ is {\em augmentable} if there exists a
$\SeqTable'$ (with a corresponding $\SeqAllocs'$) which is
\emph{equivalent} to $\SeqTable$ for all types expect $\Type_i$ and
has $\SeqAllocs(\Type_i,n) > \SeqAllocs'(\Type_i,n)$.\footnote{We
  define $\Type_0$ to be augmentable unless the dummy agent never
  retains the token in which case all agents are non-augmentable (and
  for technical reasons we declare the dummy agent to be
  non-augmentable as well).}

\begin{lemma}
\label{lem:degen}%
For any stochastic transition table $\SeqTable$ there exists an equivalent $\SeqTable'$ that is non-degenerate.
\end{lemma}

\begin{lemma}
\label{lem:aug}%
For any non-degenerate stochastic transition table $\SeqTable$, any
non-augmentable type $\Type_{i}$ always wins against any augmentable
type $\Type_{i'}$.
I.e.,
\begin{itemize}
\item if $i' < i$ and $\Type_{i'}$ has non-zero probability of holding
  the token then $\SeqTable(\Type_{i'}, \Type_{i})=1$, i.e.,
  $\Type_{i}$ always takes the token away from $\Type_{i'}$, and
\item if $i < i'$ and $\Type_{i}$ has non-zero probability of holding
  the token then $\SeqTable(\Type_{i}, \Type_{i'})=0$, i.e.,
  $\Type_{i'}$ never takes the token away from $\Type_{i}$.
\end{itemize}
\end{lemma}

It is possible to view the token passing in stochastic sequential allocation as a network flow.  From this perspective, the
augmentable and non-augmentable types form a minimum-cut and \autoref{lem:aug} states that the token must eventually flow from
the augmentable to non-augmentable types.  We defer the proof of this lemma to Appendix~\ref{app:border_1} where the main difficulty
in its proof is that the edges in the relevant flow problem have dynamic(non-constant) capacities.

\begin{proof}[Proof of Theorem~\ref{thm:seq_alloc}]
Any normalized interim allocation rule that can be implemented by the SSA algorithm is obviously feasible, so we only need to
prove the opposite direction. The proof is by contradiction, i.e., given a normalized interim allocation rule $\PreAlloc$ we show
that if there is no $(\SeqAllocs, \SeqTranses) \in \SeqAllocSpace$ such that $\PreAlloc(\cdot)=\SeqAllocs(\cdot, n)$, then
$\PreAlloc$ must be infeasible. Consider the following linear program for a given $\PreAlloc$ (i.e., $\PreAlloc$ is constant).
\begin{alignat*}{3}
    \text{maximize}     & \qquad    & &\sum_{\Type_{i} \in \TypeSpace_{\Range{1}{n}}} \SeqAllocs(\Type_{i}, n) & \qquad & \\
    \text{subject to}   &           & &\SeqAllocs(\Type_{i}, n) \le \PreAlloc(\Type_{i}), & &\forall \Type_{i} \in \TypeSpace_{\Range{1}{n}} \\
                        &           & &(\SeqAllocs, \SeqTrans) \in \SeqAllocSpace.
\end{alignat*}

Let $(\SeqAllocs, \SeqTrans)$ be an optimal assignment of this LP. If the first set of inequalities are all tight (i.e.,
$\PreAlloc(\cdot)=\SeqAllocs(\cdot, n)$) then $\PreAlloc$ can be implemented by the SSA, so by contradiction there must exists a
type $\NonAugType \in \UnionTypeSpace$ for which the inequality is not tight. Note that $\NonAugType$ cannot be augmentable
--- otherwise, by the definition of augmentability, the objective of the LP could be improved. Partition
$\TypeSpace_\Agents$ to augmentable types $\TypeSpace^+_\Agents$ and non-augmentable types $\TypeSpace^-_\Agents$. Note that
$\TypeSpace^-_\Agents$ is non-empty because $\NonAugType \in \TypeSpace^-_\Agents$. Without loss of generality, by
\autoref{lem:degen} we may assume that $(\SeqAllocs, \SeqTrans)$ is non-degenerate.\footnote{By \autoref{lem:degen}, there exits
an non-degenerate assignment with the same objective value.}

An agent wins if she holds the token at the end of the SSA algorithm. The ex ante probability that some agent with
non-augmentable type wins is $\sum_{\Type_i \in \TypeSpace^-_\Agents} \SeqAlloc(\Type_i, n)$. On the other hand,
\autoref{lem:aug} implies that the first (in the order agents are considered by SSA) agent with non-augmentable type will take
the token from her predecessors and, while she may lose the token to another non-augmentable type, the token will not be
relinquished to any augmentable type.  Therefore, the probability that an agent with a non-augmentable type is the winner is
exactly equal to the probability that at least one such agent exists, therefore
\begin{align*}
    \Prx[\Types \sim \DistProbs]{\exists i : \Type_{i} \in \TypeSpace^-_\Agents}
        &= \sum_{\Type_{i} \in \TypeSpace^-_\Agents} \SeqAlloc(\Type_{i}, n)
        < \sum_{\Type_{i} \in \TypeSpace^-_\Agents} \PreAlloc(\Type_{i}).
\end{align*}
The second inequality follows from the assumption above that $\NonAugType$ satisfies $\SeqAlloc(\NonAugType, n) <
\PreAlloc(\NonAugType)$. We conclude that $\PreAlloc$ requires an agent with non-augmentable type to win more frequently than
such an agent exists, which is a contradiction to interim feasibility of $\PreAlloc$.
\end{proof}

The contradiction that we derived in the proof of \autoref{thm:seq_alloc} yields a necessary and sufficient condition, as
formally stated in the following corollary, for feasibility of any given normalized interim allocation rule.
\begin{corollary}
\label{cor:1cons}%
A normalized interim allocation rule $\PreAlloc$ is feasible if and only if
\begin{align}
    \sum_{\TypeVar \in \TypeSubset} \PreAlloc(\TypeVar)  & \le \Prx[\Types \sim \DistProbs]{\exists i: \Type_i \in \TypeSubset}, &
        & \forall \TypeSubset \subseteq \UnionTypeSpace \tag{MRMB} \label{eq:MRMB}
\end{align}
\end{corollary}

The necessity of condition \eqref{eq:MRMB} is trivial and its sufficiency was previously proved by \citet{B91}. This condition
implies that the space of all feasible normalized interim allocation rules, $\PreAllocSpace$, can be specified by $2^\NumTypes$
linear constraints on $\NumTypes$-dimensional vectors $\PreAlloc$. An important consequence of \autoref{thm:seq_alloc} is that
$\PreAllocSpace$ can equivalently be formulated by only $O(\NumTypes^2)$ variables and  $O(\NumTypes^2)$ linear constraints as a
projection of $\SeqAllocSpace$, therefore any optimization problem over $\PreAllocSpace$ can equivalently be solved over
$\SeqAllocSpace$.

\subsection{$k$-Unit Feasibility Constraints}
\label{sec:border_k}%

In this section, we consider environments where at most $k$ agents can
be simultaneously allocated to. First, we generalize Border's
characterization of interim feasibility to environments with $k$-unit
feasibility constraint. Our generalization implies that the space of
feasible normalized interim allocation rules is a polymatroid. Second,
we observe that optimization problems can be efficiently solved over
polymatroids; this allows us to optimize over feasible interim
allocation rules. Third, we show that the normalized interim
allocation rules corresponding to the vertices of this polymatroid are
implemented by simple deterministic ordered-subset-based allocation
mechanisms. Furthermore, for any point in this polymatroid, the
corresponding normalized interim allocation rule can be implemented
by, (i) expressing it as a convex combination of the vertices of the
polymatroid, (ii) sampling from this convex combination, and (iii)
using the ordered subset mechanism corresponding to the sampled
vertex. We present an efficient randomized rounding routine for
rounding a point in a polymatroid to a vertex which combines the steps
(i) and (ii). These approaches together yield efficient algorithms for
optimizing and implementing interim allocation rules.

\paragraph{Polymatroid Preliminaries.}
Consider an arbitrary set function $\SubFun : 2^{\Universe} \to \PosReals$ defined over an arbitrary finite set $\Universe$; let
$\PolyMat{\SubFun}$ denote the polytope associate with $\SubFun$ defined as
\begin{align*}
    \PolyMat{\SubFun} &= \left\{ \PolyVec \in \PosReals^\Universe \middle| \forall \USubset \subseteq \Universe : \PolyVec(\USubset) \le \SubFun(\USubset)
    \right\}
\end{align*}
where $\PolyVec(\USubset)$ denotes $\sum_{\Elem \in \USubset}
\PolyVec(\Elem)$. The convex polytope $\PolyMat{\SubFun}$ is called a
\emph{polymatroid} if $\SubFun$ is a submodular function. Even though
a polymatroid is defined by an exponential number of linear
inequalities, the separation problem for any given $\PolyVec \in
\PosReals^{\Universe}$ can be solved in polynomial time as follows:
find $\displaystyle \USubset^* = \argmin_{\USubset}
\SubFun(\USubset)-\PolyVec(\USubset)$; if $\PolyVec$ is infeasible,
the inequality $\PolyVec(\USubset^*) \le \SubFun(\USubset^*)$ must be
violated, and that yields a separating hyperplane for $\PolyVec$. Note
that $\SubFun(\USubset)-\PolyVec(\USubset)$ is itself submodular in
$\USubset$, so it can be minimized in strong polynomial
time. Consequently, optimization problems can be solved over
polymatroids in polynomial time.  Next, we describe a characterization
of the vertices of a polymatroid. This characterization plays an
important role in our proofs and also in our ex post implementation of
interim allocation rules.

\begin{definition}[ordered subset]
\label{def:ranking}%
For an arbitrary finite set $\Universe$, an {\em ordered subset}
$\Ranking \subset \Universe$ is given by an ordering on elements
$\Ranking = (\RElem_1,\ldots,\RElem_{|\Ranking|})$ where shorthand
notation $\RElem_r \in \Ranking$ denotes the $r$th element in
$\Ranking$.
\end{definition}

\begin{proposition}
\label{prop:vertex}%
Let $\SubFun : 2^\Universe \to \PosReals$ be an arbitrary
non-decreasing submodular function with $\SubFun(\emptyset) = 0$ and
let $\PolyMat{\SubFun}$ be the associated polymatroid with the set of
vertices $\VertexSet{\PolyMat{\SubFun}}$. Every ordered subset
$\Ranking$ of $\Universe$ (see \autoref{def:ranking}) corresponds to a
vertex of $\PolyMat{\SubFun}$, denoted by
$\Vertex{\Ranking}{\PolyMat{\SubFun}}$, which is computed as
follows.  
\begin{align*}
    \forall \Elem \in \Universe :& \qquad &
    \PolyVec(\Elem) &=
        \begin{cases}
        \SubFun(\{\RElem_1,\ldots,\RElem_r\})-\SubFun(\{\RElem_1,\ldots,\RElem_{r-1}\}) & \text{if $\Elem = \RElem_r \in \Ranking$} \\
        0       & \text{if $\Elem \not \in \Ranking$}
        \end{cases}
\end{align*}
Furthermore, for every $\PolyVec \in \VertexSet{\PolyMat{\SubFun}}$
there exist a corresponding $\Ranking$.
\end{proposition}
It is easy to see that for any $\PolyVec \in
\VertexSet{\PolyMat{\SubFun}}$, an associated $\Ranking$ can be found
efficiently by a greedy algorithm (see \citet{S03} for a comprehensive
treatment of polymatroids).

\paragraph{Ordered subset allocation mechanisms.}
The following class of allocation mechanisms are of particular importance both in our characterization of interim feasibility and
in ex post implementation of interim allocation rules.

\begin{definition}[ordered subset allocation mechanism]
\label{def:rank_alloc_k} Parameterized by a \emph{ordered subset} $\Ranking$ of $\UnionTypeSpace$ (see \autoref{def:ranking}),
the {\em ordered subset mechanism}, on profile of types $\Types \in
\TypeSpaces$, orders the agents based on their types according to
$\Ranking$, and allocates to the agents greedily (e.g., with $k$ units
available the $k$ first ordered agents received a unit).  If an agent
$i$ with type $\type_i \not\in \Ranking$ is never served.\footnote{The
  virtual valuation maximizing mechanisms from the classical
  literature on revenue maximizing auctions are ordered subset mechanisms,
  see, e.g., \citet{M81}, an observation made previousl by
  \citet{elk-07}. The difference between these ordered subset mechanisms and the
  classic virtual valuation maximization mechanisms is that our
  ordered subset will come from solving an optimization on the whole auction
  problem where as the Myerson's virtual values come directly from
  single-agent optimizations.}
\end{definition}

\paragraph{Characterization of interim feasibility.}
Border's characterization of interim feasibility for $k=1$ unit
auctions states that the probability of serving a type in a subspace
of type space is no more than the probability that a type in that
subspace shows up.  This upper bound is equivalent to the expected
minimum of one and the number of types from the subspace that show up;
furthermore, this equivalent phrasing of the upper bound extends to
characterize interim feasibility in $k$-unit auctions.

Consider expressing an ex post allocation for type profile $\types$ by
$\RvAlloc{\Types} \in \{0,1\}^\UnionTypeSpace$ as follows.  For all
$\type'_i \in \UnionTypeSpace$, $\RvAlloc{\Types}(\type'_i) = 1$ if
player $i$ is served and $\type_i = \type_i'$ and 0 otherwise.  This
definition of ex post allocations is convenient as the normalized
interim allocation rule is calculated by taking its expecation, i.e.,
$\PreAlloc(\type'_i) = \expect[\types]{\RvAlloc{\Types}(\type'_i)}$.
Ex post feasibility requires that,
\begin{align}
    \RvAlloc{\Types}(\TypeSubset) & \le \min(\Abs{\Types \cap \TypeSubset}, k), &
        &  \forall \Types \in \TypeSpaces, \forall \TypeSubset \subseteq \UnionTypeSpace
        \label{eq:bk1}
\end{align}
In other words: for any profile of types $\types$, the number of types
in $\TypeSubset$ that are served by $\RvAlloc{\types}$ must be at most
the number of types in $S$ that showed up in $\types$ and the upper
bound $k$.  Taking expectations of both sides of this equation with
respect to $\types$ motivates the following definition and theorem.

\begin{definition}
\label{def:incidence_k} The {\em expected rank function} for distribution $\denss$ and subspace $S \subset \UnionTypeSpace$ is
\begin{align}
    \DistOracle_k(\TypeSubset) &= \Ex[\Types \sim
      \DistProbs]{\min\left(\Abs{\Types \cap \TypeSubset}, k\right)}
    \tag{$\DistOracle_k$} \label{eq:dist_oracle}
\end{align}
where $\Types \cap \TypeSubset$ denotes $\{\Type_1,\ldots,\Type_n\}
\cap \TypeSubset$.
\end{definition}

\begin{theorem}
\label{thm:border_k}%
For supply constraint $k$ and distribution $\denss$, the space of all feasible normalized interim allocation rules, $\PreAllocSpace$, is the polymatroid associated with
$\DistOracle_k$, i.e., $\PreAllocSpace = \PolyMat{\DistOracle_k}$, i.e., for all $\PreAlloc \in \PreAllocSpace$, 
\begin{align}
    \PreAlloc(\TypeSubset) &\le \Ex[\Types]{\min(\Abs{\Types \cap \TypeSubset}, k)} = \DistOracle_k(\TypeSubset), & \forall \TypeSubset \subseteq \UnionTypeSpace
    \label{eq:bk2}
\end{align}
\end{theorem}

The proof of this theorem will be deferred to the next section where
we will derive a more general theorem.  A key step in the proof will
be relating the statement of the theorem to the polymatroid theory
described already.  To show that the constraint of the theorem is a
polymatroid, we observe that the expected rank function is submoduler.

\begin{lemma} 
\label{lem:g_k-submodular}
The expected rank function $\DistOracle_k$ is submodular.
\end{lemma}

\begin{proof}
Observe that for any fixed $\Types$, $\min(\Types \cap
\TypeSubset, k)$ is obviously a submodular function in $\TypeSubset$, and therefore  $\DistOracle_k$ is a convex
combination\footnote{Note that taking the expectation is the same as taking a convex combination.} of submodular functions, so
$\DistOracle_k$ is submodular.
\end{proof}

We now relate vertices of the polymatroid to ordered subset allocation
mechanisms.

\begin{theorem}
\label{thm:border_k-unique}
For supply constraint $k$, if $\PreAlloc \in \PreAllocSpace$ is the vertex
$\Vertex{\Ranking}{\PolyMat{\DistOracle_k}}$ of the polymatroid
$\PolyMat{\DistOracle_k}$ the unique ex post implementation is the
ordered subset mechanism induced by $\Ranking$
(\autoref{def:rank_alloc_k}).
\end{theorem}

\begin{proof}
Let $\PreAlloc=\Vertex{\Ranking}{\PolyMat{\DistOracle_k}}$ be an
arbitrary vertex of $\PolyMat{\DistOracle_k}$ with a corresponding
ordered subset $\Ranking$; by \autoref{prop:vertex}, such a $\Ranking$
exists for every vertex of a polymatroid.  For every integer $r \leq
|\Ranking|$, define $\TypeSubset^r =\{\RElem_1,\ldots,\RElem_r\}$ as
the $r$-element prefix of the ordering. By \autoref{prop:vertex},
inequality~\eqref{eq:bk2} must be tight for every $\TypeSubset^r$
which implies that inequality~\eqref{eq:bk1} must also be tight for
every $\TypeSubset^r$ and every $\Types \in \TypeSpaces$. Observe that
inequality~\eqref{eq:bk1} being tight for a subset $\TypeSubset$ of
types implies that any ex post allocation mechanism implementing
$\PreAlloc$ must allocate as much as possible to types in
$\TypeSubset$.  By definition, an ordered subset mechanism allocates
to as many types as possible (up to $k$) from each $\TypeSubset^r$;
this is the unique outcome given that inequality~\eqref{eq:bk1} is
tight for every $\TypeSubset^r$.
\end{proof}

\paragraph{Optimization over feasible interim allocation rules.}
The characterization of interim feasibility as a polymatroid
constraint immediately enables efficient solving of optimization
problems over the feasible interim allocation rules as long as we can
compute $\DistOracle_k$ efficiently (see \citet{S03} for optimization
over polymatroids). The following lemma states that $\DistOracle_k$
can be computed efficiently.

\begin{lemma}
\label{lem:g_k}%
For independent agent (i.e., if $\DistProbs$ is a product distribution), $\DistOracle_k(\TypeSubset)$ can be exactly computed in
time $O((n+\Abs{\TypeSubset})\cdot k)$ for any $\TypeSubset \in \UnionTypeSpace$ using dynamic programming.
\end{lemma}

%

\paragraph{Ex post implementation of feasible interim allocation rules.}

We now address the task of finding an ex post implementation
corresponding to any $\PreAlloc \in \PreAllocSpace$. By
\autoref{thm:border_matroid}, if $\PreAlloc$ is a vertex of
$\PreAllocSpace$, it can be easily implemented by an ordered subset
allocation mechanism (\autoref{def:rank_alloc_k}). As any point in the
polymatroid (or any convex polytope) can be specified as a convex
combination of its vertices, to implement the corresponding interim
allocation rule it is enough to show that this convex combination can
be efficiently sampled from. An ex post implementation can then be
obtained by sampling a vertex and using the ordered subset mechanism
corresponding to that vertex. Instead of explicitly computing this
convex combination, we present a general randomized rounding routine
$\RandRound(\cdot)$ which takes a point in a polymatroid and returns a
vertex of the polymatroid such that the expected value of every
coordinate of the returned vertex is the same as the original
point. This approach is formally described next.

\begin{definition}[randomized ordered subset allocation mechanism]
\label{def:rra_k}%
Parameterized by a normalized interim allocation rule $\PreAlloc
\in \PreAllocSpace$, a \emph{randomized ordered subset allocation
  mechanism (RRA)} computes the allocation for a profile of types
$\Types \in \TypeSpaces$ as follows.
\begin{enumerate}
\item
Let $(\PreAlloc^*, \Ranking^*) \leftarrow \RandRound(\PreAlloc)$.
\item
Run the ordered subset mechanism (\autoref{def:rank_alloc_k}) with
ordered subset $\Ranking^*$.
\end{enumerate}
\end{definition}

\begin{theorem}
\label{thm:rra_k}%
Any normalized interim allocation rule $\PreAlloc \in \PreAllocSpace$
can be implemented by the randomized ordered subset allocation
mechanism (\autoref{def:rra_k}) as a distribution over deterministic
ordered subset allocation mechanisms.
\end{theorem}
\begin{proof}
The proof follows from linearity of expectation.
\end{proof}

\paragraph{Randomized rounding for polymatroids.}
We describe $\RandRound(\cdot)$ for general polymatroids. First, we
present a few definitions and give an overview of the rounding
operator. Consider an arbitrary finite set $\Universe$ and a
polymatroid $\PolyMat{\SubFun}$ associated with a non-decreasing
submodular function $\SubFun : 2^\Universe \to \PosReals$ with
$\SubFun(\emptyset) = 0$. A set $\USubset \subseteq \Universe$ is
called \emph{tight} with respect to a $\PolyVec \in
\PolyMat{\SubFun}$, if and only if
$\PolyVec(\USubset)=\SubFun(\USubset)$. A set
$\TightSets=\{\USubset^0,\ldots,\USubset^{m}\}$ of subsets of
$\Universe$ is called a \emph{nested family of tight sets} with
respect to $\PolyVec \in \PolyMat{\SubFun}$, if and only the elements
of $\TightSets$ can be or ordered and indexed such that $\emptyset =
\USubset^0 \subset \cdots \subset \USubset^m \subseteq \Universe$, and
such that $\USubset^r$ is tight with respect to $\PolyVec$ (for every
$r \in [m])$.

$\RandRound(\PolyVec)$ takes an arbitrary $\PolyVec \in
\PolyMat{\SubFun}$ and iteratively makes small changes to it until a
vertex is reached. At each iteration $\ell$, it computes
$\PolyVec^{\ell} \in \PolyMat{\SubFun}$, and a nested family of tight
sets $\TightSets^\ell$ (with respect to $\PolyVec^{\ell}$) such that
\begin{itemize}
\item
$\Ex{\PolyVec^\ell | \PolyVec^{\ell-1}}=\PolyVec^{\ell-1}$, and

\item
$\PolyVec^\ell$ is closer to a vertex (compared to $\PolyVec^{\ell-1}$) in the sense that either the number of non-zero
coordinates has decreased by one or the number of tight sets has increased by one.
\end{itemize}
Observe that the above process must stop after at most $2\Abs{\Universe}$ iterations\footnote{In fact we will show that it stops
after at most $\Abs{\Universe}$ iterations.}. At each iteration $\ell$ of the rounding process, a vector $\PolyVecDelta \in
\Reals^\Universe$ and $\delta, \delta' \in \PosReals$ are computed such that both $\PolyVec^{\ell-1}+\delta \cdot\PolyVecDelta$
and $\PolyVec^{\ell-1}-\delta'\cdot\PolyVecDelta$ are still in $\PolyMat{\SubFun}$, but closer to a vertex. The algorithm then
chooses a random $\delta'' \in \{\delta, -\delta'\}$ such that $\Ex{\delta''}=0$, and sets $\PolyVec^{\ell} \leftarrow
\PolyVec^{\ell-1}+\delta''\cdot\PolyVecDelta$.

\begin{definition}[$\RandRound(\PolyVec)$]
This operator takes as its input a point $\PolyVec \in \PolyMat{\SubFun}$ and returns as its output a pair $(\PolyVec^*,
\Ranking^*)$, where $\PolyVec^*$ is a random vertex of $\PolyMat{\SubFun}$ and $\Ranking^*$ is its associated ordered subset
(see \autoref{prop:vertex}), and such that $\Ex{\PolyVec^*}=\PolyVec$.

The algorithm modifies $\PolyVec$ iteratively until a vertex is reached. It also maintains a nested family of tight sets
$\TightSets$ with respect to $\PolyVec$. As we modify $\TightSets$, we always maintain an ordered labeling of its elements, i.e.,
if $\TightSets=\{\USubset^0,\ldots,\USubset^m\}$, we assume that $\emptyset=\USubset^0 \subset \cdots \subset \USubset^m
\subseteq \Universe$; in particular, the indices are updated whenever a new tight set is added. For each $\Elem \in \Universe$,
define $\IVec{\Elem} \in [0,1]^\Universe$ as a vector whose value is $1$ at coordinate $\Elem$ and $0$ everywhere else.

\begin{enumerate}
\item
Initialize $\TightSets \leftarrow \{\emptyset\}$.

\item
Repeat each of the following steps until no longer applicable:
\begin{itemize}
\item
If there exist distinct $\Elem, \Elem' \in \USubset^r \setminus \USubset^{r-1}$ for some $r \in [m]$:
\begin{enumerate}
\item
Set $\PolyVecDelta \leftarrow \IVec{\Elem}-\IVec{\Elem'}$, and compute $\delta,\delta' \in \PosReals$ such that
$\PolyVec+\delta\cdot \PolyVecDelta$ has a new tight set $\USubset$ and $\PolyVec-\delta'\cdot \PolyVecDelta$ has a new tight set
$\USubset'$, i.e,
\begin{itemize}
\item
set $\displaystyle \USubset \leftarrow \argmin_{\USubset^{r-1}+\Elem \subseteq \USubset \subseteq \USubset^r-\Elem'}
\SubFun(\USubset)-\PolyVec(\USubset)$, and $\delta \leftarrow \SubFun(\USubset)-\PolyVec(\USubset)$;

\item
set $\displaystyle  \USubset' \leftarrow \argmin_{\USubset^{r-1}+\Elem' \subseteq \USubset' \subseteq \USubset^r-\Elem}
\SubFun(\USubset')-\PolyVec(\USubset')$, and $\delta' \leftarrow \SubFun(\USubset')-\PolyVec(\USubset')$.
\end{itemize}

\item
$\begin{cases}
    \text{with prob. $\frac{\delta}{\delta+\delta'}$: \qquad
        set $\PolyVec \leftarrow \PolyVec+\delta\cdot\PolyVecDelta$, and add $\USubset$ to $\TightSets$.} & \\
    \text{with prob. $\frac{\delta'}{\delta+\delta'}$: \qquad
        set $\PolyVec \leftarrow \PolyVec-\delta'\cdot\PolyVecDelta$, and add $\USubset'$ to $\TightSets$.} &
\end{cases}$
\end{enumerate}

\item
If there exists $\Elem \in \Universe \setminus \USubset^m$ for which $\PolyVec(\Elem) > 0$:
\begin{enumerate}
\item
Set $\PolyVecDelta \leftarrow \IVec{\Elem}$, and compute $\delta,\delta' \in \PosReals$ such that $\PolyVec+\delta\cdot
\PolyVecDelta$ has a new tight set $\USubset$ and $\PolyVec-\delta'\cdot \PolyVecDelta$ has a zero at coordinate $\Elem$, i.e,
\begin{itemize}
\item
set $\displaystyle \USubset \leftarrow \argmin_{\USubset \supseteq \USubset^m+\Elem} \SubFun(\USubset)-\PolyVec(\USubset)$, and
$\delta \leftarrow \SubFun(\USubset)-\PolyVec(\USubset)$;

\item
set $\delta' \leftarrow \PolyVec(\Elem)$.
\end{itemize}

\item
$\begin{cases}
    \text{with prob. $\frac{\delta}{\delta+\delta'}$: \qquad
        set $\PolyVec \leftarrow \PolyVec+\delta\cdot\PolyVecDelta$, and add $\USubset$ to $\TightSets$.} & \\
    \text{with prob. $\frac{\delta'}{\delta+\delta'}$: \qquad
        set $\PolyVec \leftarrow \PolyVec-\delta'\cdot\PolyVecDelta$ } &
\end{cases}$

\end{enumerate}
\end{itemize}

\item Set $\PolyVec^* \leftarrow \PolyVec$ and define the ordered subset $\Ranking^* : \USubset^m \to [m]$ according to $\TightSets$, i.e.,
for each $r \in [m]$ and $\Elem \in \USubset^r \setminus \USubset^{r-1}$, define $\Ranking^*(\Elem) = r$.
\item Return $(\PolyVec^*,\Ranking^*)$.

\end{enumerate}
\end{definition}

\begin{theorem}
\label{thm:randround}%
For any non-decreasing submodular function $\SubFun : 2^\Universe \to \PosReals$ and any $\PolyVec \in \PolyMat{\SubFun}$, the
operator $\RandRound(\PolyVec)$ returns a random $(\PolyVec^*, \Ranking^*)$ such that $\PolyVec^* \in
\VertexSet{\PolyMat{\SubFun}}$, and $\Ranking^*$ is the ordered subset corresponding to $\PolyVec^*$ (see
\autoref{prop:vertex}), and such that $\Ex{\PolyVec^*}=\PolyVec$. Furthermore, the algorithm runs in strong polynomial time. In
particular, it runs for $O(\Abs{\Universe})$ iterations where each iteration involves solving two submodular minimizations.
\end{theorem}

\subsection{Matroid Feasibility Constraints}
\label{sec:border_matroid}%

In this section, we consider environments where the feasibility constraints are encoded by a matroid $\Mat=(\UnionTypeSpace,
\IndepSets)$. For every type profile $\Types \in \TypeSpaces$, a subset $\TypeSubset \subseteq \{\Type_1,\ldots, \Type_n\}$ can
be simultaneously allocated to if and only if $\TypeSubset \in \IndepSets$. We show that the results of \autoref{sec:border_k}
can be easily generalized to environments with matroid feasibility constraints.

\paragraph{Matroid Preliminaries.}
A matroid $\Mat=(\Universe, \IndepSets)$ consists of a ground set $\Universe$ and a family of independent sets $\IndepSets
\subseteq 2^\Universe$ with the following two properties.
\begin{itemize}
\item
For every $\IndepSet, \IndepSet'$, if $\IndepSet' \subset \IndepSet \in \IndepSets$, then $\IndepSet' \in \IndepSets$.

\item
For every $\IndepSet,\IndepSet' \in \IndepSets$, if $\Abs{\IndepSet'} < \Abs{\IndepSet}$, there exists $\Elem \in \IndepSet
\setminus \IndepSet'$ such that $\IndepSet'\cup \{\Elem\} \in \IndepSets$.
\end{itemize}

For every matroid $\Mat$, the rank function $\Rank{\Mat} : 2^\Universe \to \mathds{N}\cup\{0\}$ is defined as follows: for each
$\USubset \subseteq \Universe$, $\Rank{\Mat}(\USubset)$ is the size of the maximum independent subset of $\USubset$. A matroid
can be uniquely characterized by its rank function, i.e., a set $\IndepSet \subseteq \Universe$ is an independent set if and only
if $\Rank{\Mat}(\IndepSet) = \Abs{\IndepSet}$. A matroid rank function has the following two properties:
\begin{itemize}
\item $\Rank{\Mat}(\cdot)$ is a non-negative non-decreasing integral submodular function.
\item $\Rank{\Mat}(\USubset) \le \Abs{\USubset}$, for all $\USubset \subseteq \Universe$.
\end{itemize}
Furthermore, every function with the above properties defines a matroid.

Any set $\USubset \subseteq \Universe$ can be equivalently represented by its incidence vector $\chi^\USubset \in
\{0,1\}^\Universe$ which has a $1$ at every coordinate $\Elem \in \USubset$ and $0$ everywhere else.

\begin{proposition}
Consider an arbitrary finite matroid $\Mat=(\Universe, \IndepSets)$ with rank function $\Rank{\Mat}(\cdot)$. Let
$\PolyMat{\Rank{\Mat}}$ denote the polymatroid associated with $\Rank{\Mat}(\cdot)$ (see \autoref{sec:border_k}); the vertices of
$\PolyMat{\Rank{\Mat}}$ are exactly the incidence vectors of the independent sets of $\Mat$.
\end{proposition}

See~\citet{S03} for a comprehensive treatment of matroids.



\paragraph{Characterization of interim feasibility.}
We now generalize the characterization of interim feasibility as the
polymatroid given by the expected rank of the matroid.  From this
generalization the computational results of the preceeding section can
be extended from $k$-unit environments to matroids.


Let $\RandBits$ denote the random bits used by an ex post allocation
rule, and let $\RvAlloc{\Types,\RandBits} \in \{0,1\}^\UnionTypeSpace$
denote the ex post allocation rule (i.e., the incidence vector of the
subset of types that get allocated to) for type profile $\Types \in
\TypeSpaces$ and random bits $\RandBits$. It is easy to see that
$\RvAlloc{\Types,\RandBits}$ is a feasible ex post allocation if and
only if it satisfies the following class of inequalities.
\begin{align}
    \RvAlloc{\Types,\RandBits}(\TypeSubset) & \le \Rank{\Mat}(\Types \cap \TypeSubset), &
        &  \forall \Types \in \TypeSpaces, \forall \TypeSubset \subseteq \UnionTypeSpace
        \label{eq:bm1}
\end{align}
The above inequality states that the subset of types that get
allocated to must be an independent set of the restriction of matroid
$\Mat$ to $\{\Type_1,\ldots,\Type_n\}$.  The expectation of the
left-hand-side is exactly the normalized interim allocation rule,
i.e., for any $\type'_i \in \UnionTypeSpace$, $\PreAlloc(\type'_i) =
\expect[\Types,\RandBits]{\RvAlloc{\Types,\RandBits}(\type'_i)}$.
Taking expectations of both sides of \eqref{eq:bm1} then motivates the
following definition and theorem that characterize interim
feasibility.

\begin{definition}
\label{def:incidence_matroid} The {\em expected rank} for distribudion $\denss$, subspace $S \subset \UnionTypeSpace$, and matroid $\Mat$ with rank function $\Rank{\Mat}$ is
\begin{align}
    \DistOracle_\Mat(\TypeSubset) &= \Ex[\Types \sim
      \DistProbs]{\Rank{\Mat}(\Types \cap \TypeSubset)}
    \tag{$\DistOracle_\Mat$} \label{eq:dist_oracle_matroid}
\end{align}
where $\Types \cap \TypeSubset$ denotes $\{\Type_1,\ldots,\Type_n\} \cap \TypeSubset$.
\end{definition}

\begin{theorem}
\label{thm:border_matroid}%
For matroid $\Mat$ and distribution $\denss$,
the space of all feasible normalized interim allocation rules, $\PreAllocSpace$, is the polymatroid associated with
$\DistOracle_k$, i.e., $\PreAllocSpace = \PolyMat{\DistOracle_\Mat}$, i.e., for all $\PreAlloc \in \PreAllocSpace$, 
\begin{align}
    \PreAlloc(\TypeSubset) &\le \Ex[\Types]{\Rank{\Mat}(\Types \cap \TypeSubset)} = \DistOracle_\Mat(\TypeSubset), & \forall \TypeSubset \subseteq \UnionTypeSpace
    \label{eq:bm2}
\end{align}
\end{theorem}

\begin{theorem}
\label{thm:border_matroid-unique}
For matroid $\Mat$, if $\PreAlloc \in \PreAllocSpace$ is the vertex
$\Vertex{\Ranking}{\PolyMat{\DistOracle_\Mat}}$ of the polymatroid
$\PolyMat{\DistOracle_\Mat}$ the unique ex post implementation is the
ordered subset mechanism induced by $\Ranking$
(\autoref{def:rank_alloc_k}).
\end{theorem}

To prove the above theorems, we use the following decomposition lemma which applies to general polymatroids.

\begin{lemma}[Polymatroidal Decomposition]
\label{lem:decompose}%
Let $\Universe$ be an arbitrary finite set, $\SubFun^1,\ldots,\SubFun^m : 2^\Universe \to \PosReals$ be arbitrary non-decreasing
submodular functions, and $\SubFun^*=\sum_{j=1}^m \lambda^j \SubFun^j$ be an arbitrary convex combination of them. For every
$\PolyVec^*$ the following holds: $\PolyVec^* \in \PolyMat{\SubFun^*}$ if and only if  it can be decomposed as $\PolyVec^* =
\sum_{j=1}^m \lambda^j \PolyVec^j$ such that $\PolyVec^j \in \PolyMat{\SubFun^j}$ (for each $j \in [m]$). Furthermore, if
$\PolyVec^*$ is a vertex of $\PolyMat{\SubFun^*}$, this decomposition is unique. More precisely, if
$\PolyVec^*=\Vertex{\Ranking}{\PolyMat{\SubFun^*}}$ for some ordered subset $\Ranking$, then
$\PolyVec^j=\Vertex{\Ranking}{\PolyMat{\SubFun^j}}$ (for each $j \in [m])$.
\end{lemma}
\begin{proof}
First, observe that the only-if part is obviously true, i.e., if $\PolyVec^j \in \PolyMat{\SubFun^j}$ (for each $j \in [m]$), we
can write
\begin{align}
    \PolyVec^j(\USubset) & \le \SubFun^j(\USubset) & & \forall \USubset \subseteq \Universe, \label{eq:pmd1}
\intertext{multiplying both sides by $\lambda^j$ and summing over all $j\in [m]$ we obtain}
    \PolyVec^*(\USubset)=\sum_{i=1}^m \lambda^j \PolyVec^j(\USubset) & \le \sum_{j=1}^m \lambda^j \SubFun^j(\USubset)=\SubFun^*(\USubset)  & & \forall \USubset \subseteq
    \Universe,
    \label{eq:pmd2}
\end{align}
which implies that $\PolyVec^* \in \PolyMat{\SubFun^*}$.

Next, we prove that for every $\PolyVec^* \in \PolyMat{\SubFun^*}$
such a decomposition exists. Note that a polymatroid is a convex
polytope, so any $\PolyVec^* \in \PolyMat{\SubFun^*}$ can be written
as a convex combination of vertices as $\PolyVec^*=\sum_{\ell}
\alpha^\ell {\PolyVec^*}^\ell$, where each ${\PolyVec^*}^\ell$ is a
vertex of $\PolyMat{\SubFun^*}$; consequently, if we prove the claim
for the vertices of $\PolyMat{\SubFun^*}$, i.e., that
${\PolyVec^*}^\ell=\sum_{j=1}^m \lambda^j \PolyVec^{\ell,j}$ for some
$\PolyVec^{\ell,j} \in \PolyMat{\SubFun^j}$, then a decomposition of
$\PolyVec^* = \sum_{j=1}^m \lambda^j \PolyVec^j$ can be obtained by
setting $\PolyVec^j = \sum_\ell \alpha^\ell \PolyVec^{\ell,j}$.

Next, we prove the second part of the theorem which also implies that
a decomposition exists for every vertex of $\PolyMat{\SubFun^*}$. Let
$\PolyVec^*=\Vertex{\Ranking}{\PolyMat{\SubFun^*}}$ be an arbitrary
vertex of $\PolyMat{\SubFun^*}$ with a corresponding ordered subset
$\Ranking$; by \autoref{prop:vertex}, such a $\Ranking$ exists for
every vertex of a polymatroid. For every integer $r \leq |\Ranking|$,
define $\TypeSubset^r =\{\RElem_1,\ldots,\RElem_r\}$ as the
$r$-element prefix of the ordering.  By \autoref{prop:vertex},
inequality~\eqref{eq:pmd2} is tight for every $\USubset^r$, which
implies that inequality~\eqref{eq:pmd1} must also be tight for each
$\USubset^r$ and for every $j \in [m]$. Consequently, for each $r \in [|\Ranking|]$, and each $j \in [m]$, by taking the difference
of the inequality~\eqref{eq:pmd2} for $\USubset^{r}$ and
$\USubset^{r-1}$, given that they are tight, we obtain
\begin{align*}
    \PolyVec^j(\Elem) & =\SubFun^j(\USubset^{r})-\SubFun^j(\USubset^{r-1})
\end{align*}
Furthermore, for each $\Elem \not \in \Ranking$, $\PolyVec^*(\Elem)=0$ which implies that $\PolyVec^j(\Elem)=0$ for
every $j \in [m]$. Observe that we have obtained a unique $\PolyVec^j$ for each $j \in [m]$ which is exactly the vertex of
$\PolyMat{\SubFun^j}$ corresponding to $\Ranking$ as described in \autoref{prop:vertex}. It is easy to verify that indeed
$\PolyVec^* = \sum_{j=1}^m \lambda^j \PolyVec^j$.
\end{proof}

\begin{proof}[Proof of \autoref{thm:border_matroid}]
The inequality in equation \eqref{eq:bm1} states that the subset of types that get allocated to must be an independent set of the restriction of
matroid $\Mat$ to $\{\Type_1,\ldots,\Type_n\}$. Define $\Rank[\Types ]{\Mat}(\TypeSubset) = \Rank{\Mat}(\Types \cap \TypeSubset)$
(for all $(\TypeSubset \subseteq \UnionTypeSpace$). Notice that $\Rank[\Types ]{\Mat}$ is a submodular function. The above
inequality implies that $\RvAlloc{\Types,\RandBits} \in \PolyMat{\Rank[\Types ]{\Mat}}$. Define $\RvAlloc{\Types} =
\Ex[\RandBits]{\RvAlloc{\Types,\RandBits}}$. Observe that $\PreAlloc=\Ex[\Types]{\RvAlloc{\Types}}$, so $\PreAlloc$ is a
feasible normalized interim allocation rule if and only if it can be decomposed as $\PreAlloc = \sum_{\Types \in \TypeSpaces}
\DistProb(\Types) \RvAlloc{\Types}$ where $\RvAlloc{\Types} \in \PolyMat{\Rank[\Types]{\Mat}}$ for every $\Types \in
\TypeSpaces$; by \autoref{lem:decompose}, this is equivalent to $\PreAlloc \in \PolyMat{\DistOracle_{\Mat}}$ where
$\DistOracle_\Mat(\TypeSubset) = \sum_{\Types \in \TypeSpaces} \DistProb(\Types) \Rank[\Types]{
\Mat}(\TypeSubset)=\Ex[\Types]{\Rank[\Types]{\Mat}(\TypeSubset)}$ (for all $\TypeSubset \subseteq \UnionTypeSpace$), as defined
in \autoref{def:incidence_matroid}. That completes the proof of the first part of the theorem.
\end{proof}

\begin{proof}[Proof of \autoref{thm:border_matroid-unique}]
Suppose $\PreAlloc=\Vertex{\Ranking}{\PolyMat{\DistOracle_\Mat}}$ for
some ordered subset $\Ranking$. By \autoref{lem:decompose}, the
decomposition of $\PreAlloc$ is unique and is given by
$\RvAlloc{\Types} =
\Vertex{\Ranking}{\PolyMat{\Rank[\Types]{\Mat}}}$. Notice that this is
the same allocation obtained by the deterministic rank-based
allocation mechanism which ranks according to $\Ranking$ (see
\autoref{def:rank_alloc_k}).
\end{proof}

\paragraph{Optimization over feasible interim allocation rules.}

As in \autoref{sec:border_k}, the characterization of interim
feasibility as a polymatroid constraint immediately enables efficient
solving of optimization problems over the feasible normalized interim
allocation rules as long as we can compute $\DistOracle_\Mat$
efficiently (see \citet{S03} for optimization over
polymatroids). Depending on the specific matroid, it might be possible
to exactly compute $\DistOracle_\Mat$ in polynomial time (e.g., as in
\autoref{lem:g_k}); otherwise, it can be computed approximately within
a factor of $1-\epsilon$ and with probability $1-\delta$, by sampling,
in time polynomial in $\frac{1}{\epsilon}$ and $\frac{1}{\delta}$.

\paragraph{Ex post implementation of feasible interim allocation rules.}
An ex post implementation for any $\PreAlloc \in \PreAllocSpace$ can
be obtained exactly as in \autoref{sec:border_k}. 


\begin{corollary}
\label{thm:rra_matroid}%
Any normalized interim allocation rule $\PreAlloc \in \PreAllocSpace$
can be implemented by the randomized rank-based allocation mechanism
(\autoref{def:rra_k}) as a distribution over deterministic
rank-based allocation mechanisms (\autoref{def:rank_alloc_k}).
\end{corollary}

\section{Conclusions and Extensions}

\label{sec:conc}


In this paper we have focused on binary allocation problems where an
agent is either served or not served.  For these binary allocation
problems distributions over allocations are given by a single number,
i.e., the probability that the agent is served.  Our results can be
extended to environments with multi-unit demand when the agents
utility is linear in the expected number of units the agent receives.

In Section~\ref{sec:inter} we described algorithms for
optimizing over feasible interim allocation rules and for (ex post)
implementation of the resulting rules.  Neither these algorithms nor
the generalization of Border's condition require the types of the
agents to be independently distributed.  However, our formulation of
incentive compatibility for interim allocation rules does require
independence.  For correlated distributions the interim allocation
rule is a function of the actual type of the agent (which conditions
the types of the other agents) and the reported type of the agent.
Therefore, this generalization of our theorem to correlated
environments has little relevance for mechanism design.

The algorithms in Section~\ref{sec:inter} do not require the
feasibility constraint to be known in advance.  A simple example where
this generalization is interesting is a multi-unit auction where the
supply $k$ is stochastically drawn from a known distribution.  Our
result shows that the optimal auction in such an environment can be
described by picking the random ordering on types and allocating
greedily by this ordering while supplies last.  We do not know of many
examples other than this where this generalization is interesting.

Our techniques can also be used in conjunction with the approach of
\citet{CDW11} for solving multi-item auction problems for agents with
additive values.

One important extension of our work is to scenarios where the type
space and distribution are only available via oracle access (and can
be very large or even infinite).  Given a polynomial time
approximation scheme for a variant of the single agent problem we can
construct a polynomial time approximation scheme for the multi-agent
problem.  Such a model is important, for instance, when the agents
type space is multi-dimensional but succinctly describable, e.g., for
unit demand agents with independently distributed values for various
items for sale.  Such a type space would be exponentially large in the
number of items but succinctly described in polynomial space in the
number of items.  While reduction can be applied to this scenario;
however, we do not know of any solution to the optimal single-agent
problem with which to instantiate the reduction.


\bibliographystyle{apalike}
\bibliography{bibs}

\appendix
\section{Proofs from Section~\ref{sec:border_1}}
\label{app:border_1}%

We first describe a network flow formulation of $\SeqAllocSpace$, which is used to prove \autoref{lem:degen} and
\autoref{lem:aug}.

\begin{figure}
\center
\includegraphics[width=.95\textwidth]{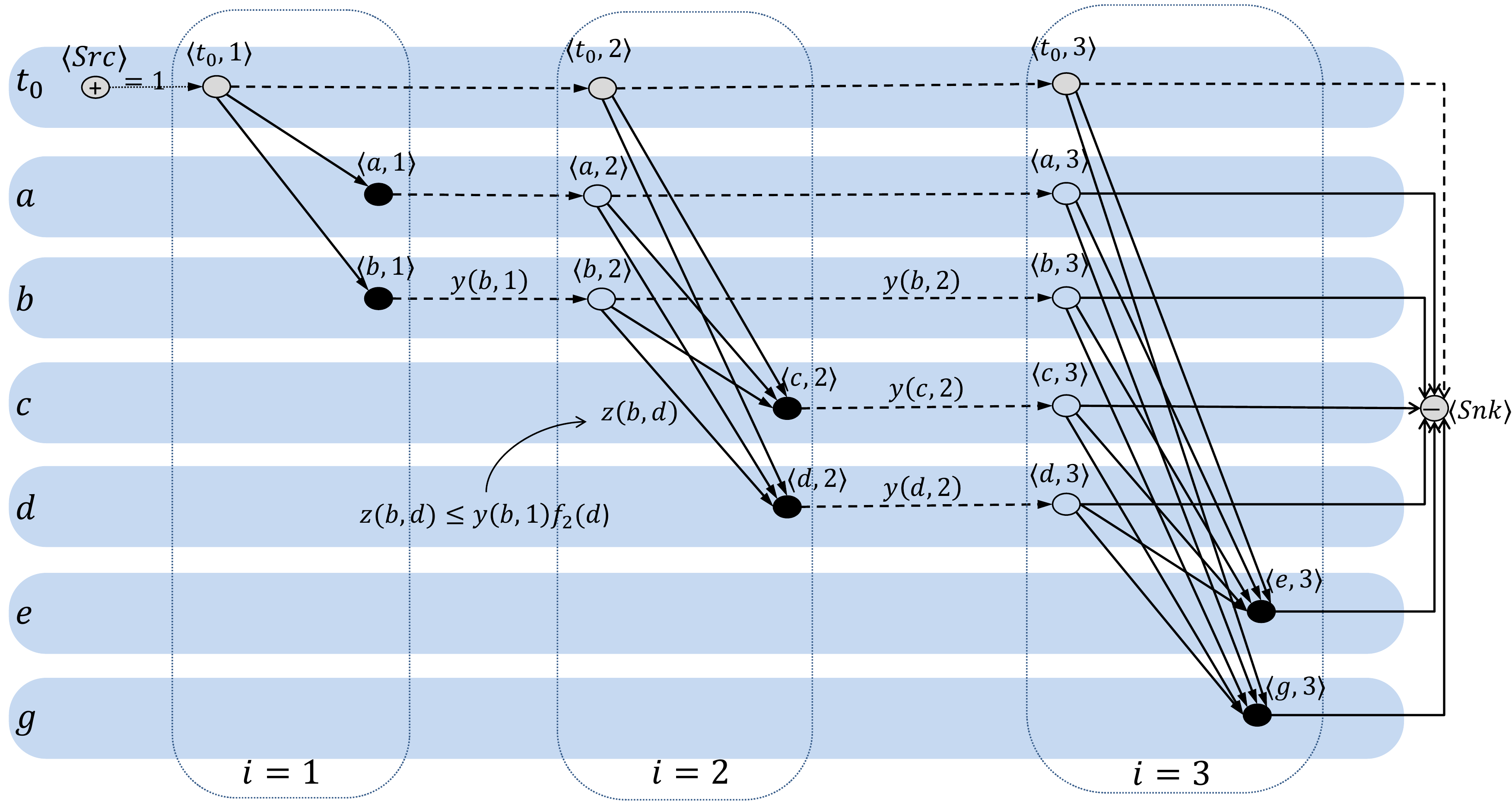}
\caption{The flow network corresponding to the SSA algorithm~\ref{def:ssa}. In this instance, there are three agents with type
spaces $\TypeSpace_1 = \{a,b\}$, $\TypeSpace_2=\{c,d\}$, and $\TypeSpace_3=\{e,g\}$. All nodes in the same row correspond to the
same type. The diagonal edges have dynamic capacity constraints while all other edges have no capacity constraints. The flow
going from $\Node{\Type_{i'}}{i}$ to $\Node{\Type_{i}}{i}$ corresponds to the ex-ante probability of $\Type_{i}$ taking the token
away from $\Type_{i'}$. The flow going from $\Node{\Type_{i'}}{i}$ to $\Node{\Type_{i'}}{i+1}$ corresponds to the ex-ante
probability of $\Type_{i'}$ still holding the token after agent $i$ is visited. \label{fig:seq_alloc}}
\end{figure}

\paragraph{A network flow formulation of $\SeqAllocSpace$.}
We construct a network in which every feasible flow corresponds to some $(\SeqAllocs, \SeqTranses) \in \SeqAllocSpace$. The
network (see \autoref{fig:seq_alloc}) has a source node $\Src$, a sink node $\Snk$, and $n-i+1$ nodes for every $\Type_i \in
\UnionTypeSpace$ labeled as $\Node{\Type_i}{i},\cdots,\Node{\Type_i}{n}$ where each node $\Node{\Type_{i'}}{i}$ corresponds to
the type $\Type_{i'}$ at the time SSA algorithm is visiting agent $i$. For each $\Type_{i'} \in \UnionTypeSpace$ and for each $i
\in \Range{i'}{n-1}$ there is an edge $(\Node{\Type_{i'}}{i}, \Node{\Type_{i'}}{i+1})$ with infinite capacity whose flow is equal
to $\SeqAlloc(\Type_{i'}, i)$; we refer to these edges as \emph{``horizontal edges''}. For every $\Type_{i'}$ and every
$\Type_{i}$ where $i' < i$ there is an edge $(\Node{\Type_{i'}}{i}, \Node{\Type_{i}}{i})$ whose flow is equal to
$\SeqTrans(\Type_{i'}, \Type_{i})$ and whose capacity is equal to the total amount of flow that enters $\Node{\Type_{i'}}{i}$
multiplied by $\DistProb_{i}(\Type_{i})$, i.e., it has a dynamic capacity which is equal to $\SeqAlloc(\Type_{i'}, i-1)
\DistProb_{i}(\Type_{i})$; we refer to these edges as \emph{``diagonal edges''}. There is an edge $(\Src, \Type_0)$ through which
the source node pushes exactly one unit of flow. Finally, for every $\Type_i \in \UnionTypeSpace$, there is an edge
$(\Node{\Type_i}{n}, \Snk)$ with unlimited capacity whose flow is equal to $\SeqAlloc(\Type_i, n)$. To simplify the proofs we
sometimes use $\Node{\Type_0}{0}$ as an alias for the source node $\Src$ and $\Node{\Type_i}{n+1}$ as aliases for the sink node
$\Snk$. The network always has a feasible flow because all the flow can be routed along the path $\Src, \Node{\Type_0}{1},\ldots,
\Node{\Type_0}{n}, \Snk$.

We define the \emph{residual capacity} between two types $\Type_{i'}, \Type_{i} \in \UnionTypeSpace$ with respect to a given
$(\SeqAllocs, \SeqTranses) \in \SeqAllocSpace$ as follows.
\begin{align*}
    \RCap_{\SeqAllocs, \SeqTranses}(\Type_{i'}, \Type_{i})=
        \begin{cases}
            \SeqAlloc(\Type_{i'}, i-1)\DistProb_{i}(\Type_{i})-\SeqTrans(\Type_{i'}, \Type_{i}) & i > i' \\
            \SeqTrans(\Type_{i}, \Type_{i'})                                                    & i < i' \\
            0                                                                                   & \text{otherwise}
        \end{cases} \tag{$\RCap$} \label{eq:rcap}
\end{align*}

Due to dynamic capacity constraints, it is not possible to augment a flow along a path with positive residual capacity by simply
changing the amount of the flow along the edges of the path, because reducing the total flow entering a node also decreases the
capacity of the diagonal edges leaving that node, which could potentially violate their capacity constraints. Therefore, we
introduce an operator $\Reroute(\Type_{i'}, \Type_{i}, \rho)$ (algorithm~\ref{alg:reroute} and \autoref{fig:reroute}) which
modifies an existing $(\SeqAllocs, \SeqTrans) \in \SeqAllocSpace$, while maintaining its feasibility, to transfer a
$\rho$-fraction of $\SeqAllocs(\Type_i, n)$ to $\SeqAllocs(\Type_{i'}, n)$ by changing the flow along the cycle
\begin{align*}
\Snk,\Node{\Type_{i'}}{n},\Node{\Type_{i'}}{n-1},\ldots,\Node{\Type_{i'}}{\max(i',i)},\Node{\Type_{i}}{\max(i',i)},\ldots,\Node{\Type_{i}}{n-1},\Node{\Type_{i}}{n},
\Snk
\end{align*}
and adjusting the flow of the the diagonal edges which leave this cycle. More precisely, $\Reroute(\Type_{i'}, \Type_{i}, \rho)$
takes out a $\rho$-fraction of the flow going through the subtree rooted at $\Node{\Type_{i'}}{\max(i',i)}$~\footnote{This
subtree consists of the path $\Node{\Type_{i'}}{\max(i',i)},\ldots, \Node{\Type_{i'}}{n}, \Snk$ and all the diagonal edges
leaving this path.} and reassigns it to the subtree rooted at $\Node{\Type_{i}}{\max(i',i)}$ (see \autoref{fig:reroute}).

\begin{algorithm}[h]
\textbf{Input:} An existing $(\SeqAllocs, \SeqTranses) \in \SeqAllocSpace$ given implicitly, a source type $\Type_{i'}
\in \UnionTypeSpace$, a destination type $\Type_{i} \in \UnionTypeSpace$ where $i' \neq i$, and a fraction $\rho \in [0,1]$.\\
\textbf{Output:} Modify $(\SeqAllocs, \SeqTranses)$  to transfer a $\rho$-fraction of $\SeqAllocs(\Type_{i'}, n)$ to
$\SeqAllocs(\Type_{i}, n)$ while ensuring that the modified assignment is still in $\SeqAllocSpace$.
\begin{algorithmic}[1]
    \IF{$i' < i$}
        \STATE Increase $\SeqTrans(\Type_{i'}, \Type_{i})$ by $\rho\cdot \SeqAlloc(\Type_{i'},i)$.
    \ELSE
        \STATE Decrease $\SeqTrans(\Type_{i},\Type_{i'})$ by $\rho\cdot \SeqAlloc(\Type_{i'},i')$.
    \ENDIF
    \FOR{$i''=\max(i',i)$ to $n$}
        \STATE Increase $\SeqAlloc(\Type_{i},i'')$ by $\rho\cdot\SeqAlloc(\Type_{i'},i'')$.
        \STATE Decrease $\SeqAlloc(\Type_{i'}, i'')$ by $\rho\cdot\SeqAlloc(\Type_{i'},i'')$.
    \ENDFOR
    \FOR{$\Type_{i''} \in \TypeSpace_{\Range{\max(i',i)+1}{n}}$}
        \STATE Increase $\SeqTrans(\Type_{i}, \Type_{i''})$ by $\rho\cdot\SeqTrans(\Type_{i'}, \Type_{i''})$.
        \STATE Decrease $\SeqTrans(\Type_{i'}, \Type_{i''})$ by $\rho\cdot\SeqTrans(\Type_{i'}, \Type_{i''})$.
    \ENDFOR
\end{algorithmic}
\caption{$\Reroute(\Type_{i'},\Type_{i}, \rho)$.\label{alg:reroute}}%
\end{algorithm}

\begin{figure}[h]
\center
\includegraphics[width=.95\textwidth]{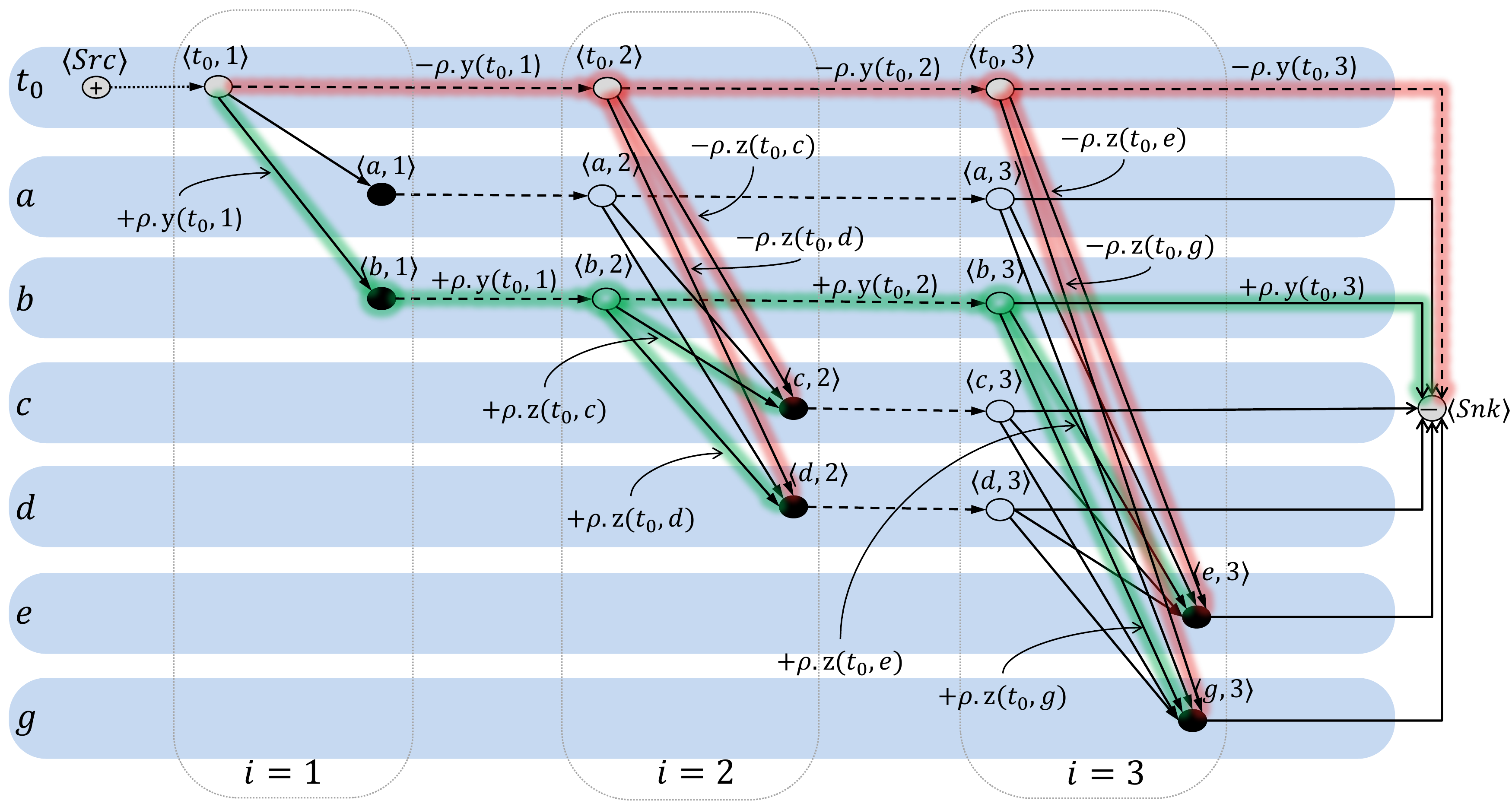}
\caption{Changes made by applying $\Reroute(\Type_0, b, \rho)$. A $\rho$-fraction of the red subtree rooted at $\Type_0$ is take
out and reassigned to the green subtree rooted at $b$. The exact amount of change is indicated for each green and each red edge.
The flow along all other edges stay intact. The operator has the effect of reassigning $\rho$-fraction of ex-ante probability of
allocation for type $\Type_0$ to type $b$. \label{fig:reroute}}
\end{figure}

\begin{proof}[Proof of Lemma~\ref{lem:degen}]
For any given $(\SeqAllocs, \SeqTrans) \in \SeqAllocSpace$ we show that it is always possible to modify $\SeqAllocs$ and
$\SeqTrans$ to obtain a non-degenerate feasible assignment with the same induced interim allocation probabilities (i.e., the same
$\SeqAllocs(\cdot, n)$). Let $d$ denote the number of degenerate types with respect to $(\SeqAllocs, \SeqTrans)$, i.e., define
\begin{align*}
    d = \#\left\{ \Type_i \in \TypeSpace_{\Range{1}{n}} \middle| \SeqAllocs(\Type_i,n)=0, \SeqAlloc(\Type_i, i) > 0\right\}
\end{align*}
The proof is by induction on $d$. The base case is $d=0$ which is trivial. We prove the claim for $d > 0$ by modifying
$\SeqAllocs$ and $\SeqTrans$, reducing the number of degenerate types to $d-1$, and then applying the induction hypothesis. Let
$\Type_{i}$ be a degenerate type. For each $\Type_{i'} \in \TypeSpace_{\Range{0}{i-1}}$, we apply the operator $\Reroute(\Type_i,
\Type_{i'}, \frac{\SeqTrans(\Type_{i'}, \Type_i)}{\SeqAlloc(\Type_i, i)})$ unless $\SeqAlloc(\Type_i, i)$ has already reached
$0$. Applying this operator to each type $\Type_{i'}$ eliminates the flow from $\Node{\Type_{i'}}{i}$ to $\Node{\Type_i}{i}$, so
eventually $\SeqAlloc(\Type_i, i)$ reaches $0$ and $\Type_i$ is no longer degenerate and also no new degenerate type is
introduced, so the number of degenerate types is reduced to $d-1$. It is also easy to see that $\SeqAlloc(\Type_{i'},n)$ is not
modified because $\SeqAlloc(\Type_{i},n) = 0$. That completes the proof.
\end{proof}

\begin{proof}[Proof of Lemma~\ref{lem:aug}]
To prove the lemma it is enough to show that for any augmentable type $\Type_{i'}$ and any non-augmentable type $\Type_{i}$,
$\RCap_{\SeqAllocs, \SeqTranses}(\Type_{i'}, \Type_{i}) =0$ which is equivalent to the statement of the lemma (the equivalence
follows from the definition of $\RCap$ and equation~\eqref{eq:seq_pi}). The proof is by contradiction. Suppose $\Type_{i'}$ is
augmentable and $\RCap_{\SeqAllocs, \SeqTranses}(\Type_{i'}, \Type_{i}) = \delta$ for some positive $\delta$; we show that
$\Type_{i}$ is also augmentable. Since $\Type_{i'}$ is augmentable, there exists a $(\SeqAllocs',\SeqTranses') \in
\SeqAllocSpace$ such that $\SeqAlloc'(\TypeVar,n)=\SeqAlloc(\TypeVar,n)$ for all $\TypeVar \in \TypeSpace_\Agents \setminus
\{\Type_0, \Type_{i'}\}$ and $\SeqAlloc'(\Type_{i'},n)-\SeqAlloc(\Type_{i'},n)=\epsilon > 0$. Define
\begin{align*}
    (\SeqAllocs'', \SeqTranses'') &= (1-\alpha)\cdot(\SeqAllocs, \SeqTranses)+\alpha\cdot(\SeqAllocs', \SeqTranses')
\end{align*}
where $\alpha \in [0,1]$ is a parameter that we specify later. Note that in $(\SeqAllocs'', \SeqTranses'')$, $\Type_{i'}$ is
augmented by $\alpha\epsilon$, and $\RCap_{\SeqAllocs'', \SeqTranses''}(\Type_{i'}, \Type_{i}) \ge (1-\alpha)\delta$, and
$(\SeqAllocs'', \SeqTranses'') \in \SeqAllocSpace$ because it is a convex combination of $(\SeqAllocs, \SeqTranses)$ and
$(\SeqAllocs', \SeqTranses')$. Consider applying $\Reroute(\Type_{i'}, \Type_{i}, \rho)$ to $(\SeqAllocs'', \SeqTranses'')$ for
some parameter $\rho \in [0,1]$. The idea is to choose $\alpha$ and $\rho$ such that the exact amount, by which $\Type_{i'}$ was
augmented, gets reassigned to $\Type_{i}$, by applying $\Reroute(\Type_{i'}, \Type_{i}, \rho)$; so that eventually $\Type_{i}$ is
augmented while every other type (except $\Type_0$) has the same allocation probabilities as they originally had in $(\SeqAllocs,
\SeqTranses)$. It is easy to verify that by setting
\begin{align*}
    \alpha&=\frac{\SeqAlloc(\Type_{i'},n)\delta}{2} & \rho&= \frac{\epsilon\delta}{2+\epsilon\delta}
\end{align*}
we get a feasible assignment in which the allocation probability of $\Type_{i}$ is augmented by $\alpha\epsilon$ while every
other type (except $\Type_0$) has the same allocation probabilities as in $(\SeqAllocs, \SeqTranses)$. We still need to show that
$\alpha
> 0$. The proof is again by contradiction. Suppose $\alpha = 0$, so it must be $\SeqAllocs(\Type_{i'},n)=0$, which would imply that
$\Type_{i'}$ is a degenerate type because $\SeqAlloc(\Type_{i'}, i') > 0$ (because $\RCap_{\SeqAllocs,\SeqTranses}(\Type_{i'},
\Type_{i})> 0$), however $(\SeqAllocs, \SeqTranses)$ is a non-degenerate assignment by the hypothesis of the lemma, which is a
contradiction. That completes the proof.
\end{proof}

\section{Proofs from Section~\ref{sec:border_k}}
\label{app:border_k}

\begin{figure}[h]
\centering
\includegraphics[width=0.7\textwidth]{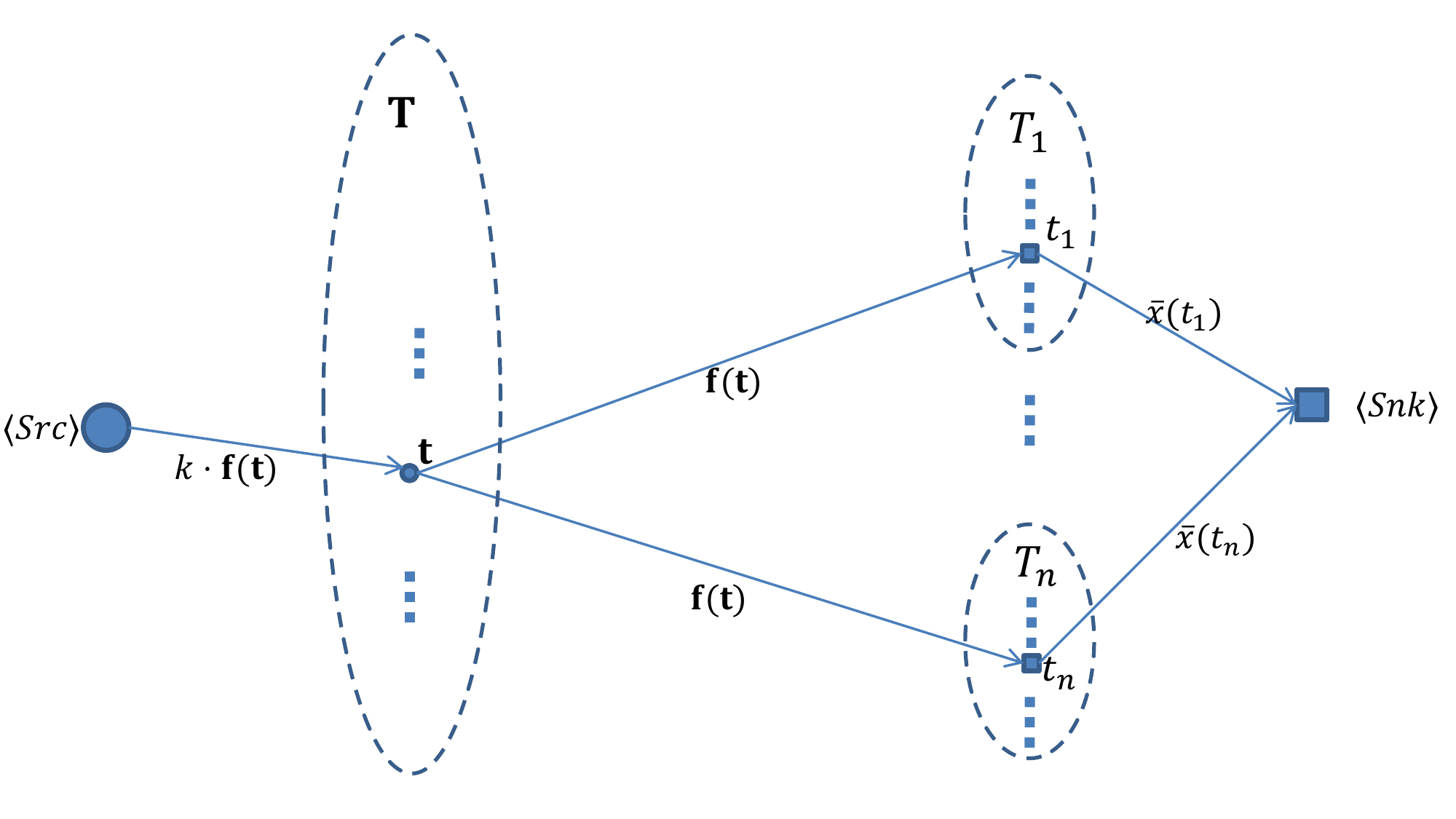}
\caption{The bipartite graph used in the max-flow/min-cut argument of the proof of \autoref{thm:border_k}. The capacities are
indicated on the edges.\label{fig:kflow}}
\end{figure}

\begin{proof}[Rest of the proof of \autoref{thm:border_k}]
We give a proof of $\PolyMat{\DistOracle_k} \subseteq \PreAllocSpace $ based on the min-cut/max-flow theorem. We start by
constructing a directed bipartite graph as illustrated in \autoref{fig:kflow}. On one side we put a node $\SNode{\Types}$, for
each type profile $\Types \in \TypeSpaces$. On the other side we put a node $\SNode{\Type_i}$, for each type $\Type_i \in
\TypeSpace_{\Range{1}{n}}$. We also add a source node $\Src$ and a sink node $\Snk$. We add a directed edge from $\Src$ to the
node $\SNode{\Types}$, for each $\Types \in \TypeSpaces$ and set the capacity of this edge to $k\cdot\DistProbs(\Types)$. We also
add $n$ outgoing edges for every node $\SNode{\Types}$, each one going to one of the nodes
$\SNode{\Type_1},\ldots,\SNode{\Type_n}$ and with a capacity of $\DistProbs(\Types)$. Finally we add a directed edge from the
node $\SNode{\Type_i}$, for each $\Type_i \in \TypeSpace_{\Range{1}{n}}$, to $\Snk$ with capacity of $\PreAlloc(\Type_i)$.
Consider a maximum flow from $\Src$ to $\Snk$. It is easy to see that there exists a feasible ex post implementation for
$\PreAlloc$ if and only if all the edges to the sink node $\Snk$ are saturated. In particular, if $\rho(\Types,\Type_i)$ denotes
the amount of flow from $\SNode{\Types}$ to $\SNode{\Type_i}$, a feasible ex post implementation can be obtained by allocating to
each type $\Type_i$ with probability $\rho(\Types,\Type_i)/\DistProbs(\Types)$ when the type profile $\Types$ is reported by the
agents.

We show that if a feasible ex post implementation does not exist, then $\PreAlloc \not \in \PolyMat{\DistOracle_k}$. Observe that
if a feasible ex post implementation does not exist, then some of the incoming edges of $\Snk$ are not saturated by the max-flow.
Let $(A, B)$ be a minimum cut such that $\Src \in A$ and $\Snk \in B$. Let $B'=B \cap \UnionTypeSpace$. We show that the
polymatroid inequality
\begin{align}
    \PreAlloc(B') &\le \DistOracle_k(B') \label{eq:bk3}
\end{align}
must have been violated. It is easy to see that the size of the cut is given by the following equation.
\begin{align*}
    \Cut(A, B) &=
        \sum_{\Types \in \TypeSpaces \cap A}\#\left\{i \middle| \Type_i \in B\right\}\DistProbs(\Types)
            +\sum_{\Types \in \TypeSpaces\cap B}k\cdot\DistProbs(\Types)+\sum_{\TypeVar \in \UnionTypeSpace\cap A}\PreAlloc(\TypeVar)
\end{align*}
Observe that for each $\Types \in \TypeSpaces \cap A$, it must be that $\#\{i | \Type_i \in B\} \le k$, otherwise moving
$\SNode{\Types}$ to $B$ would decrease the size of the cut. So the size of the minimum cut can be in simply written as:
\begin{align*}
    \Cut(A, B) &=
        \sum_{\Types \in \TypeSpaces} \min\left(\#\left\{i \middle| \Type_i \in B\right\},k\right) \DistProbs(\Types)+\sum_{\TypeVar \in
        \UnionTypeSpace\cap A}\PreAlloc(\TypeVar)
\end{align*}
On the other hand, since some of the incoming edges of $\Snk$ are not saturated by the max-flow, it must be that
\begin{align*}
    \sum_{\TypeVar\in \UnionTypeSpace} \PreAlloc(\TypeVar) &= \Cut(A\cup B-\Snk, \Snk) > \Cut(A, B),
\intertext{so}
    \sum_{\TypeVar \in \UnionTypeSpace \cap B}\PreAlloc(\TypeVar) &>
        \sum_{\Types \in \TypeSpaces} \min\left(\#\left\{i \middle| \Type_i \in B\right\},k\right) \DistProbs(\Types).
\end{align*}
The right hand side of the above inequality is the same as $\Ex[\Types \sim \DistProbs]{\min(\#\{i | \Type_i \in B\}, k)}$ which
shows that polymatroid inequality~\eqref{eq:bk3} of $\PolyMat{\DistOracle_k}$ is violated so $\PreAlloc \not \in
\PolyMat{\DistOracle_k}$. That completes the proof.
\end{proof}

\section{Proofs from Section~\ref{sec:border_matroid}}

\begin{proof}[Proof of \autoref{lem:g_k}]
Assuming that agents are independent (i.e., assuming $\DistProbs(\cdot)$ is a product distribution), $\DistOracle_k(\TypeSubset)$
can be computed in time $O((n+\Abs{\TypeSubset})\cdot k)$ using the following dynamic program in which
$\DistDP{i}{j}{\TypeSubset}$ denotes the probability of the event that $\min(\Abs{\Types \cap \TypeSubset \cap
\TypeSpace_{\Range{1}{i}}}, k)=j$.
\begin{align*}
    \DistOracle_k(\TypeSubset) &= \sum_{j=1}^k j\cdot \DistDP{n}{j}{\TypeSubset} \\
    \DistDP{i}{j}{\TypeSubset} &=
        \begin{cases}
            \DistDP{i-1}{k}{\TypeSubset}+(\sum_{\Type_i \in \TypeSubset \cap \TypeSpace_i} \DistProb_i(\Type_i))
                \cdot \DistDP{i-1}{k-1}{\TypeSubset} &  1 \le i \le n, j = k\\
            \DistDP{i-1}{j}{\TypeSubset}+(\sum_{\Type_i \in \TypeSubset \cap \TypeSpace_i} \DistProb_i(\Type_i))
                \cdot (\DistDP{i-1}{j-1}{\TypeSubset}-\DistDP{i-1}{j}{\TypeSubset}) & 1 \le i \le n, 0 \le j < k\\
            1 & i = 0, j = 0 \\
            0 & \text{otherwise}
        \end{cases} \\
\end{align*}
\end{proof}

\end{document}